\documentclass[12pt]{article}
\usepackage[margin=26mm]{geometry}

\usepackage{amsthm}

\usepackage{amsmath}
\usepackage{amssymb}
\usepackage{graphicx}
\usepackage{algorithm}
\usepackage{url}
\usepackage{epsfig}
\usepackage{color}
\usepackage{multirow}
\usepackage{wrapfig}

\newcommand{\R}{\mathbb{R}}
\newcommand{\Z}{\mathbb{Z}}
\renewcommand{\S}{\mathbb{S}}
\newcommand{\eps}{\epsilon}
\newcommand{\orient}{O}

\newcommand{\True}{\texttt{True}}
\newcommand{\False}{\texttt{False}}

\newcommand{\partition}{\mathcal{C}}

\usepackage{amsthm}

\newtheorem{theorem}{Theorem}
\newtheorem*{embtheorem}{Theorem 1}
\newtheorem{lemma}[theorem]{Lemma}

\title{%
  Embedding the dual complex\\
  of hyper-rectangular partitions
  }

\author{%
  Michael~Kerber%
  \thanks{Institute of Science and Technology Austria (IST Austria), Klosterneuburg, Austria, 
          Stanford University, Stanford, CA,
          Max-Planck-Insitute for Visual Computing and Communication, Saarbr\"ucken, Germany,  
          \texttt{mkerber@mpi-inf.mpg.de}}
}

\begin{document}
\maketitle

\begin{abstract}
A rectangular partition is the partition of an (axis-aligned) rectangle
into interior-disjoint rectangles. We ask whether a rectangular
partition permits a ``nice'' drawing of its dual, 
that is, a straight-line embedding of it 
such that each dual vertex is placed into the rectangle that it represents.
We show that deciding whether such a drawing exists is NP-complete. 
Moreover, we consider the drawing where a vertex is placed in the center
of the representing rectangle and consider sufficient conditions for this drawing
to be nice. This question is studied both in the plane 
and for the higher-dimensional generalization of rectangular partitions.
\end{abstract}

\section{Introduction}

\begin{wrapfigure}[13]{r}{3.0cm}
\vspace{-0.5cm}
\includegraphics[width=3cm]{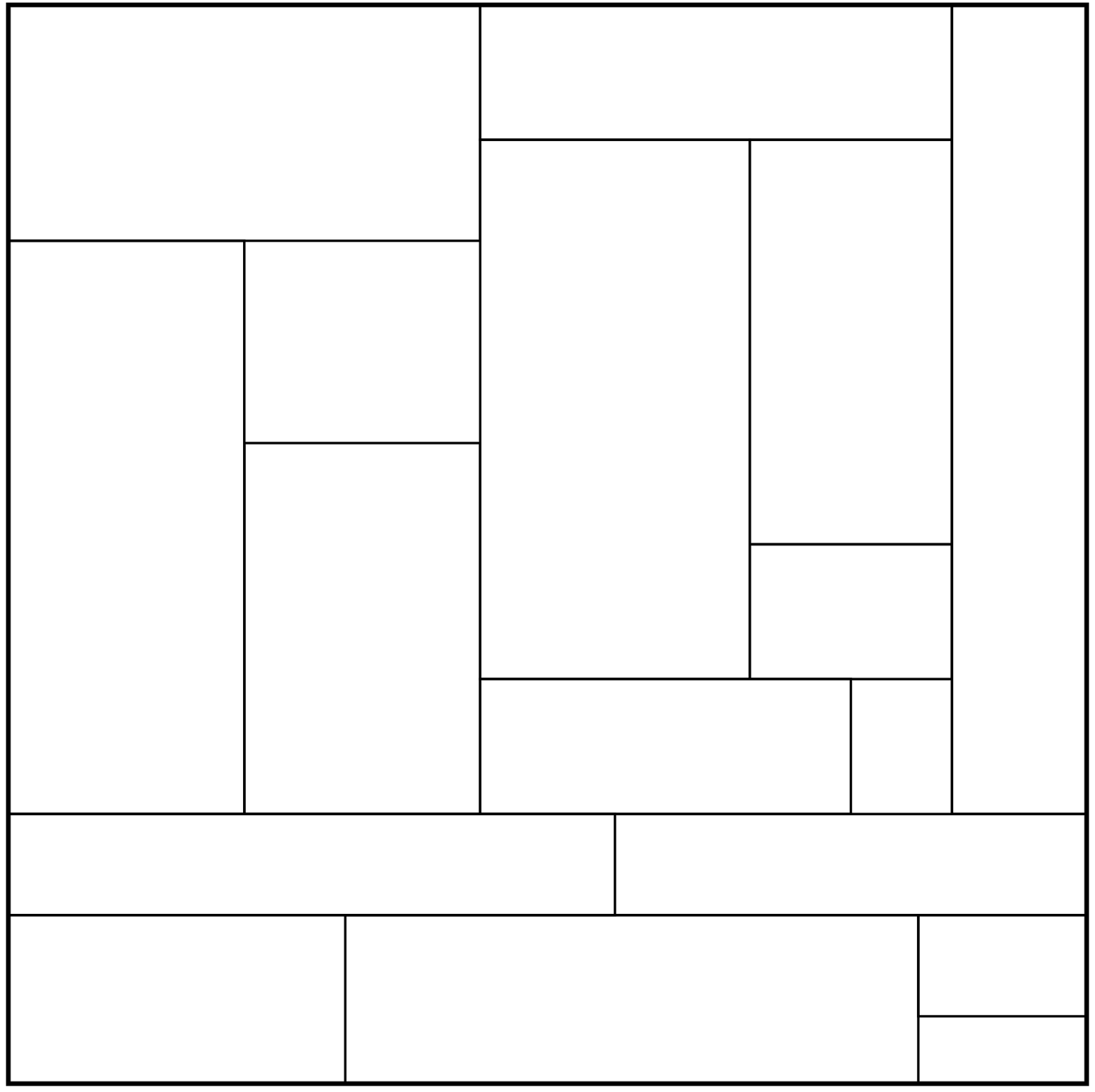}
\includegraphics[width=3cm]{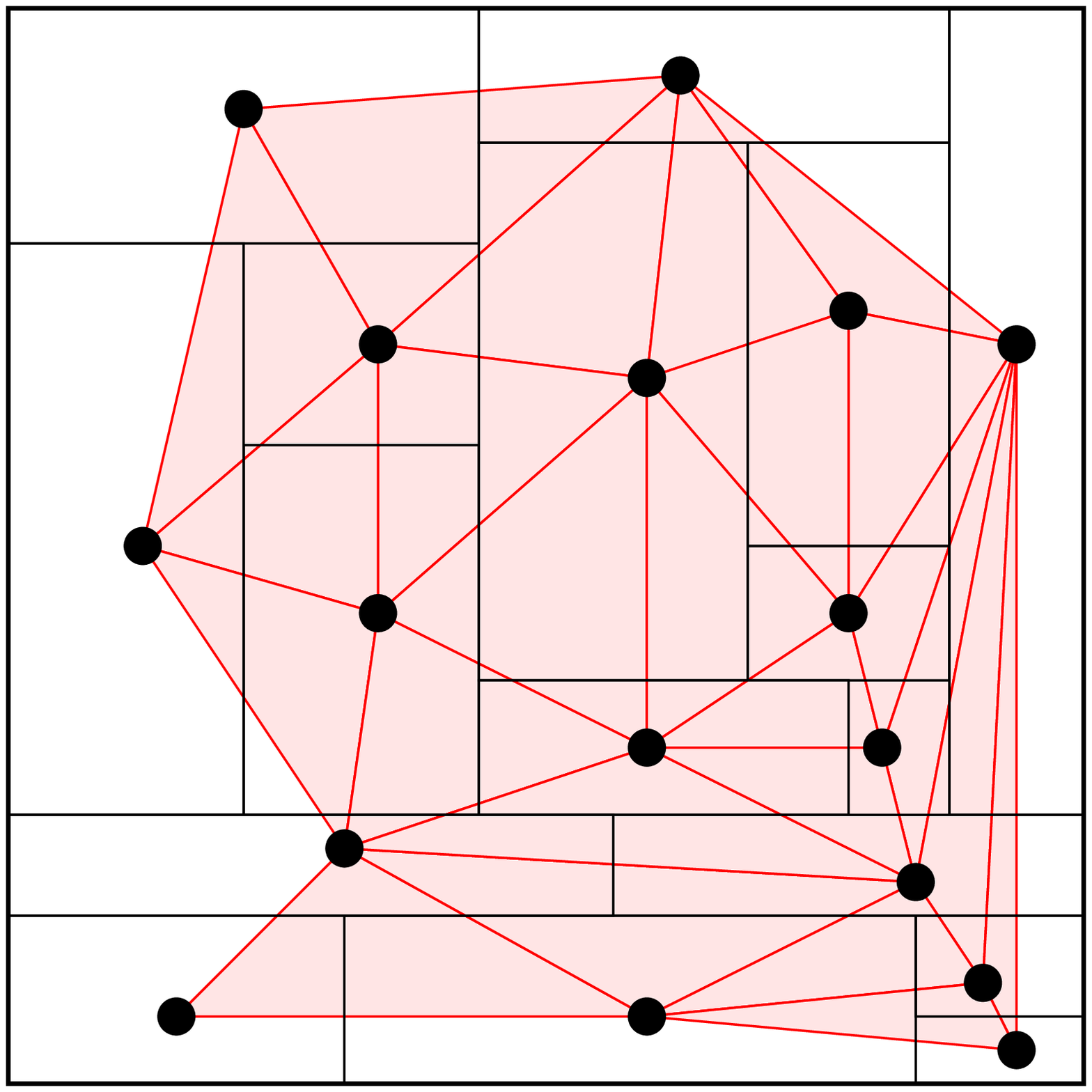}
\end{wrapfigure}
\paragraph{Problem statement and results}
We consider the partition of the $d$-dimensional cube into axis-parallel 
$d$-dimensional hyper-rectangles (or boxes)
with disjoint interiors which we call \emph{hyper-rectangular partitions}. 
The dual complex of such a partition is an abstract simplicial complex of dimension $d$
and represents the connectivity of the boxes in the partition, possibly after a slight distortion
if more than $d+1$ boxes meet in a common point. Each vertex of the dual represents
a box of the partition. 

We pose the question of whether a partition permits a ``nice'' drawing of its dual in $\R^d$. 
We impose three conditions for niceness: First of all, the drawing must be straight, that is,
a face must be drawn as the convex hull of the vertices that are on its boundary. That means that
the drawing is completely determined by the positions of the vertices. Second, we disallow
improper intersections of faces; for instance, edges are not allowed to cross except at their 
endpoints. The first two conditions can be summarized to that the drawing must yield a simplicial
complex in $\R^d$. The final condition is that each vertex must be placed in the box that it
represents; that guarantees some minimal amount of geometric closeness between
a primary cell of the partition and its dual counterpart.
If a drawing with these three conditions is possible, we call it an \emph{embedding}.
See the figure on the right for an illustration of a rectangular partition in $\R^2$ and an embedding of its dual.

We will see that not every hyper-rectangular partition permits an embedding of the dual,
so a natural question to ask is how fast can the existence of an embedding 
be decided for a given partition.
Our first main contribution is to show that this decision problem 
is NP-hard already in the planar case. More precisely, we restrict to the case that
the boxes of a partition and the vertices of the dual
are snapped to a uniform grid, and we prove that the problem
is NP-complete under this constraint; 
see Section~\ref{sec:NP} for the precise statement.

The most natural choice for placing a vertex of the dual (for arbitrary dimension $d$)
is the center of the corresponding box. 
We call a partition \emph{center-embeddable} if this vertex placement extends
to an embedding of the dual.
For instance, if all boxes are $d$-dimensional cubes of same size,
the partition is center-embeddable because its dual is simply a triangulation
of the uniform grid defined by the box centers.
We ask for simple, but less restrictive 
sufficient conditions on a partition to be center-embeddable.
We call a partition \emph{$\beta$-balanced} (with $\beta\geq 1$), if for any intersecting pair of
boxes, the ratio of the longest side of the boxes
divided by the shortest side of the boxes is at most $\beta$.
As our second main contribution, we investigate the relationships between $\beta$-balanced
and center-embeddable partitions:
In the planar case, we show that center embeddability is guaranteed 
if the partition is $(3-\eps)$-balanced, for any $\eps>0$. However, in $\R^3$ (and higher dimensions)
we can construct a $(1+\eps)$-balanced partition that is not center-embeddable. 
The situation changes if we restrict to \emph{cubical partitions}, consisting
only of $d$-dimensional cubes: In this case, $(3-\eps)$-balanced partitions are center-embeddable
in $\R^3$, and we can construct a cubical $\beta$-balanced partition in $\R^d$ 
that is not center-embeddable for any $\beta>d/(d-2)$. 
It remains open whether cubical $\beta$-balanced partitions with $\beta<d/(d-2)$
are center-embeddable in dimensions larger than $3$.

\paragraph{Related work}
Rectangular partitions in $\R^2$ have a lot of applications, for instance
 in VLSI design~\cite{lly-compact}, cartography~\cite{reisz-statistical}, and database-related applications~\cite{mps-rectangular}; 
we refer to the survey by Eppstein~\cite{eppstein-regular}
for more examples. In this context, rectangular partitions sometimes appear as the \emph{rectangular dual} of a triangulated planar graph. 
Linear-time algorithms have been presented for the computation 
of such a rectangular dual~\cite{he-finding}.
Our situation, however, is different as the rectangular partition is the input object and we ask for a straight-line embedding of its dual
where the vertices are constrained to lie inside rectangles. 
This can be seen as an instance of planar graph embedding with constraints;
NP-hardness has been shown for other constraints, such as when fixing the length of each edge~\cite{ew-fixed}, and, more related to our approach, when restricting
the placement of each vertex to a disc~\cite{godau-difficulty}. Furthermore, the problem of simultaneously embedding a planar graph and its dual on the integer grid
has received some interest, e.g.~\cite{ek-simultaneous}. Most related to our approach is the variant where the embedding of the primal graph is fixed and an embedding
of the dual is sought for, such that only primal-dual pairs of edges intersect. Bern and Gilbert~\cite{bg-drawing} show that the problem is linear-time solvable
if all faces are convex and four-sided, and becomes NP-hard for convex five-sided faces. The latter is proven with a reduction from planar 3SAT and is
similar to the proof presented in this work. However, although rectangular partitions can be seen as planar graphs with convex faces, there is no direct
reduction from our problem because we allow the dual graph to intersect the partition arbitrarily, not just at primal-dual pairs.

The higher dimensional equivalent of rectangular partitions (and their dual complexes) apparently has not been investigated 
from a theoretical point of view.
Our motivation for this topic is originated in the approximation and simplification of $d$-dimensional image data~\cite{bek-computing}. 
The idea is to identify rectangular regions in which the image looks ``similar''. In the simplest case, similarity means that the
image values inside a region are similar, but different measures can be defined. The regions define a hyper-rectangular partition
which can consist of substantially less elements than the original $n^d$ image cells. 
In some situations, a piecewise linear approximation of the image is required, e.g., for computing a level set of the image.
Standard techniques like bilinear interpolation cannot be applied because the non-uniformity of the partition leads to discontinuities;
see~\cite{wf-diamond} for a discussion with references.
The standard approach is to triangulate
the rectangular regions separately, such that triangulations of adjacent regions agree on their common boundary.
This results in many simplices for regions with a large number of neighbors.
Our dual complex construction constitutes an alternative to this standard approach; however,
it does not necessarily embed in the ambient space nicely, which leads to the question considered in this work.

A special case of rectangular partitions are \emph{hierarchical cubical subdivisions};
they arise from the initial box by a sequence of subdivisions, where a box in the partition is replaced
with $2^d$ sub-boxes of half the side length. In $2$ and $3$ dimensions, these subdivisions are called \emph{quad-trees}
and \emph{oct-trees}, respectively. A hierarchical cubical subdivisions is called \emph{balanced}
if adjacent boxes differ at most by a factor of $2$ in side length.
Note that balanced subdivisions are special cases of $2$-balanced partitions in our notation.
In~\cite{bek-computing},  the dual complex of an oct-tree is used to approximate the \emph{persistent homology} 
(see~\cite{eh-computational} for a definition) of the underlying image. 
It was also shown that the dual complex is center-embeddable, provided that the oct-tree is balanced.
In~\cite{ek-freudenthal}, that result is generalized to general hierarchical cubical subdivisions;
our results show that the hierarchical structure is crucial for obtaining this result,
as for all $d\geq 5$, we can construct a non-hierarchical $2$-balanced partition that is not center-embeddable.

\paragraph{Outline}
We introduce hyper-rectangular partitions and their dual complexes formally in Section~\ref{sec:definitions}.
Section~\ref{sec:NP} is devoted to the NP-completeness proof of finding an embedding in the planar case.
We study center-embeddability in Section~\ref{sec:balancing}. Section~\ref{sec:conclusion} concludes
the paper.

\section{Hyper-Rectangular Partitions and Dual Complexes}
\label{sec:definitions}

This section introduces the most important concepts needed for the results
of this work. 

We call a point set of the form $[a_1,b_1]\times\ldots\times [a_d,b_d]$
with $a_i<b_i$ and $a_1,\ldots,a_d,b_1,\ldots,b_d\in\Z$ an \emph{integral hyper-rectangle},
or just a \emph{box}, with \emph{lengths} $b_1-a_1,\ldots,b_d-a_d$.
An \emph{integral hyper-square}, or \emph{square box} 
is a box where all lengths are equal.
We can think of a box to be composed out of \emph{hyper-pixels}, 
or \emph{unit boxes} which are integer translates of $[0,1]^d$.
Let $B=[0,n]^d$ be a square box with an arbitrary $n>0$.
A \emph{hyper-rectangular partition} $\partition=(R_1,\ldots,R_m)$ (with $m\leq n^d$) 
of $B$ is a collection of boxes $R_i$ such that their union equals $B$ and 
their interiors are disjoint. If the union of the $R_i$ is only a subset of $B$, 
we call the collection a \emph{partial hyper-rectangular partition}. 
We usually omit the term ``hyper-rectangular'' and just talk about a \emph{partition}.
A partition is called \emph{generic} if not more than $d+1$ boxes 
intersect in a common point.

\begin{figure}
\begin{center}
\includegraphics[width=12cm]{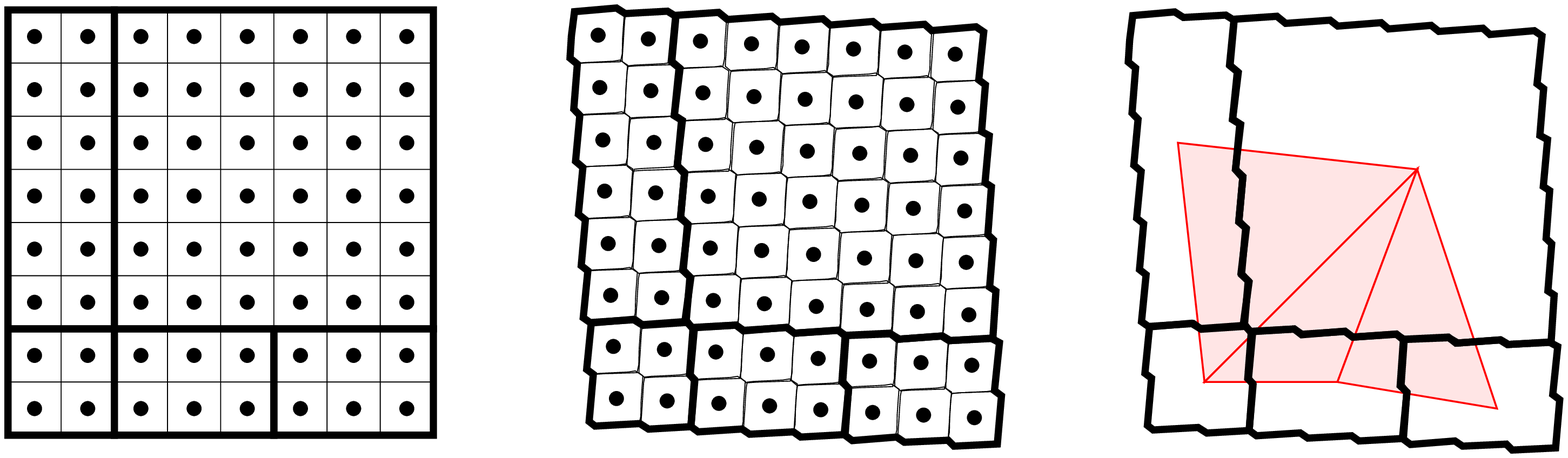}
\end{center}
\caption{Left: A non-generic rectangular partition with 5 rectangles. Middle: 
The distorted pixels and rectangles. Right: The dual complex of the partition.}
\label{fig:distortion}
\end{figure}

We associate an abstract $d$-dimensional \emph{dual complex} $D(\partition)$ to a partition $\partition$.
For a generic partition, $D(\partition)$ is just the \emph{nerve} of $\partition$, that is, each box $R_i$
is represented by a vertex $v_i$, and the simplices of $D(\partition)$ correspond to subsets
of $\{R_1,\ldots,R_m\}$ with non-empty intersection. This construction fails for non-generic
partitions because intersections of $d+2$ rectangles give rise to a $(d+1)$-simplex in the nerve.
To circumvent this problem, we slightly perturb the boxes to obtain a generic
situation, following the construction from~\cite{ek-freudenthal}: The union of unit boxes can
be seen as the Voronoi diagram of the pixel centers 
$P:=\{(x_1+1/2,\ldots,x_d+1/2) \mid x_1,\ldots,x_d\in\Z\}$.
For an arbitrary $\eps\in (0,1)$, we define the \emph{distortion} of a point $p=(x_1,\ldots,x_d)\in P$ as
\[T_\eps p := (x_1-\eps\frac{\Sigma x_i}{2},\ldots,x_d-\eps\frac{\Sigma x_i}{2}).\]
We consider the Voronoi diagram of the distorted pixel centers. 
For a pixel $U$, we define its distortion $\tilde{U}$ as the Voronoi cell of its 
distorted pixel center. For a box consisting
of pixels $U_1,\ldots,U_k$, we define its distortion
as the union of the distorted pixels $\tilde{U_1},\ldots,\tilde{U_k}$. Finally, we define
the dual complex $D(\partition)$ of a partition $\partition$ to be the nerve of the distorted boxes
in the partition. See Figure~\ref{fig:distortion} for an example in the plane. We remark that
the dual complex does not depend on the choice of $\eps$.
Informally, the distortion is a way to remove high-dimensional simplices from the nerve in
non-generic situations, thereby ``preferring'' connections between hyper-rectangle in the
diagonal direction $(1,\ldots,1)$; we refer to~\cite{ek-freudenthal} for more details.

For a hyper-rectangular partition $\partition=(R_1,\ldots,R_k)$, let $v_i$ denote
the vertex of the dual complex $D(\partition)$ that represents $R_i$. We call a mapping from
$\{v_1,\ldots,v_k\}$ to $\R^d$ a \emph{projection}. 
The \emph{barycentric refinement} of the integer grid is the grid whose vertices
are of the form $(a_1/2,\ldots,a_d/2)$ with $a_1,\ldots,a_d\in\Z$.
A projection is \emph{half-integral} if each $v_i$ is mapped to a vertex of the barycentric
refinement of the integer grid.
A projection is \emph{faithful} if each $v_i$ is mapped in the interior of $R_i$.
We assume in this work that projections are half-integral and faithful, unless otherwise stated.
A projection $\pi$ extends to the whole dual complex $D(\partition)$ by mapping a higher-dimensional simplex
to the interior of the convex hull of the projected boundary vertices. 
Abusing notation, we let $\pi$ also denote the extended mapping.
We call a projection an \emph{embedding},
if this mapping is injective, or in other words, if the image 
is a simplicial complex in $\R^d$. 

\paragraph{Seed configurations}
Consider a $d$-simplex $\sigma$ of $D(\partition)$, dual to the intersection point $w$
of the boxes $R_0,\ldots,R_d$.
By definition of $D(\partition)$, there is a \emph{unique} collection of unit boxes
$U_0,\ldots,U_d$ with $U_i\subseteq R_i$ such that the distorted
unit boxes $\tilde{U}_0,\ldots,\tilde{U}_d$ intersect \emph{in a common point} as well.
By definition, the $U_i$ intersect in $w$.
Let $u_i$ be the center of $U_i$. We call $(u_0,\ldots,u_d)$
the \emph{seed configuration} of $\sigma$.

In $\R^d$, a sequence of $d+1$ points
$p_0=(x_{0,1},\ldots,x_{0,d}),\ldots,p_d=(x_{d,1},\ldots,x_{d,d})$ 
has the \emph{orientation} 
\[
\orient(p_0,\ldots,p_d)=\mathrm{sign}\ \mathrm{det}
\left(
\begin{array}{cccc}
1 & x_{0,1} & \ldots & x_{0,d}\\
\vdots & & \ddots & \\
1 & x_{d,1} & \ldots & x_{d,d}
\end{array}
\right).
\]
\begin{wrapfigure}[10]{r}{4cm}
\vspace{-0.3cm}
\includegraphics[width=4cm]{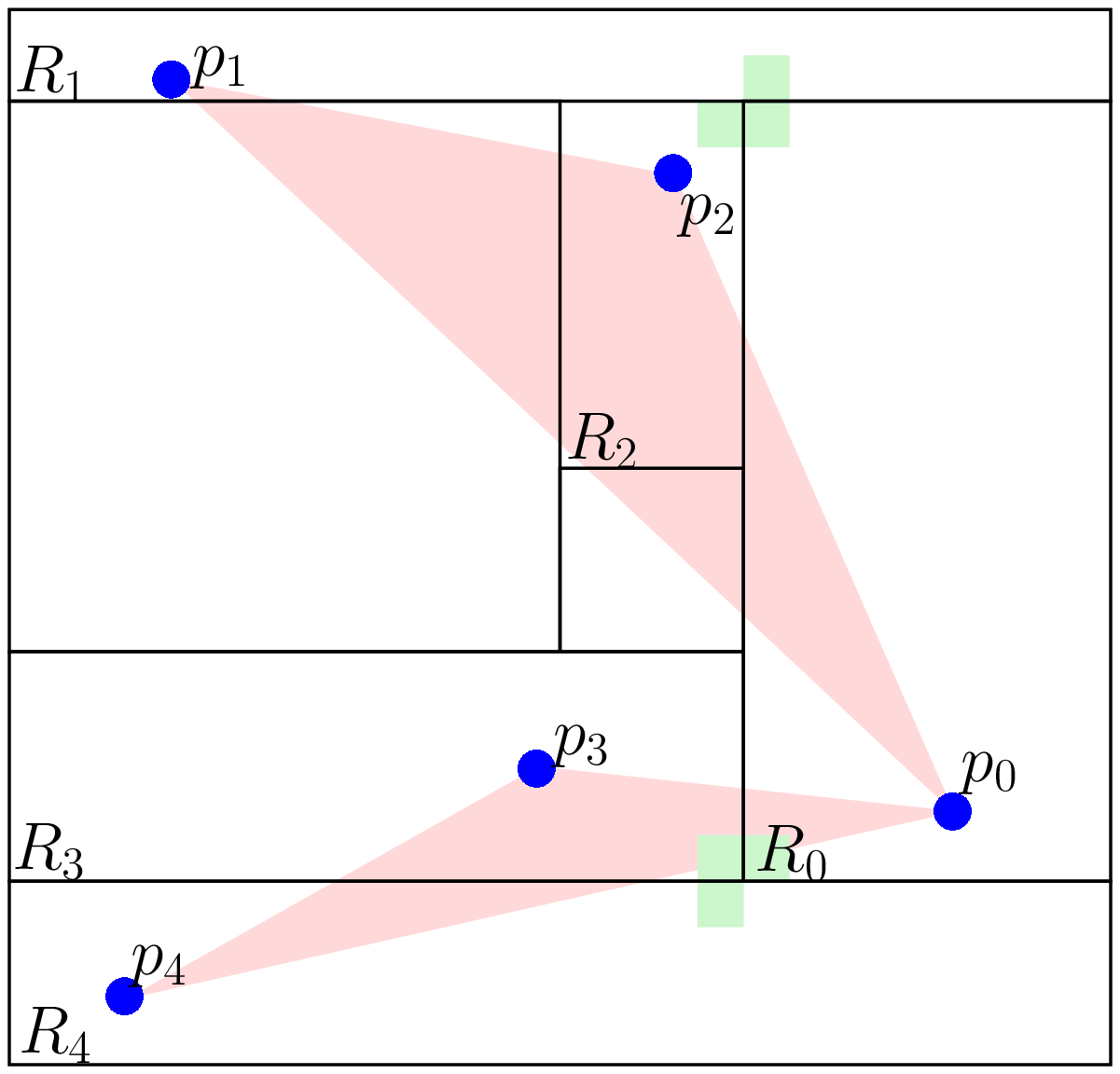}
\end{wrapfigure}
We say that a projection $\pi$ of $D(\partition)$ to $\R^d$ \emph{preserves the orientation
of the $d$-simplex $\sigma$}, if $\orient(\pi(v_0),\ldots,\pi(v_d))=\orient(u_0,\ldots,u_d)$,
where $u_0,\ldots,u_d$ is the seed configuration of $\sigma$.
For an equivalent formulation in the plane, assume that $R_0$, $R_1$, $R_2$ are arranged 
counterclockwisely
around $w$. Then $\pi$ preserves the orientation of $\sigma$, if the cycle
$\pi(v_0), \pi(v_1), \pi(v_2)$ is counterclockwisely arranged, too.
In the picture on the right, the seed configuration of the dual $2$-simplices 
$\{R_0,R_1,R_2\}$ and $\{R_0,R_3,R_4\}$ is illustrated; note that the given projection
does not preserve the orientation of the former $2$-simplex, but does preserve the orientation of the latter.

The following result is a generalization of the ``Geometric Realization Theorem''
from~\cite{ek-freudenthal}. In there, it was proven that for balanced hierarchical cubical subdivisions,
the projection that maps a vertex to the center of the cube is always an embedding. The main property
exploited in the proof is the ``Orientation Lemma'', stating that the orientation of each $d$-simplex
is preserved for such subdivisions. This is no longer true for arbitrary partitions, but the same
proof idea can be used to show:

\begin{theorem}[Embedding Theorem]\label{thm:embedding}
Let $\partition$ be a partition.
Let $\pi$ be a projection of $D(\partition)$ that preserves
the orientation of at least one $d$-simplex.
Then, $\pi$ is an embedding if and only if 
it preserves the orientation of each $d$-simplex.
\end{theorem}

The proof in~\cite{ek-freudenthal} relies on some topological concepts;
Appendix~\ref{app:embedding} repeats the argument for the convenience of the reader.
The constraint that $\pi$ preserves at least one $d$-simplex
is required to rule out pathological cases such as a partition in $\R^2$
with only $3$ rectangles, whose dual is always an embedding regardless of the orientation
of the dual triangle (unless the triangles degenerate to a line).
All our results are eventually reduced to investigate
the orientation of projected $d$-simplices: Embedability results will be proved
by showing that each $d$-simplex preserves orientation, non-embedability results
by constructing examples where some $d$-simplex does not preserve orientation.
In all these constructed examples, it will be easy to verify
that at least one $d$-simplex preserves orientation; therefore we will tacitly
ignore that assumption when applying Theorem~\ref{thm:embedding}.

\section{NP-completeness}
\label{sec:NP}

In this section, we concentrate on the planar case $d=2$. We adapt our notation
by leaving out the prefix ``hyper-'' on all defined terms.

\begin{figure}
\begin{center}
\includegraphics[width=6cm]{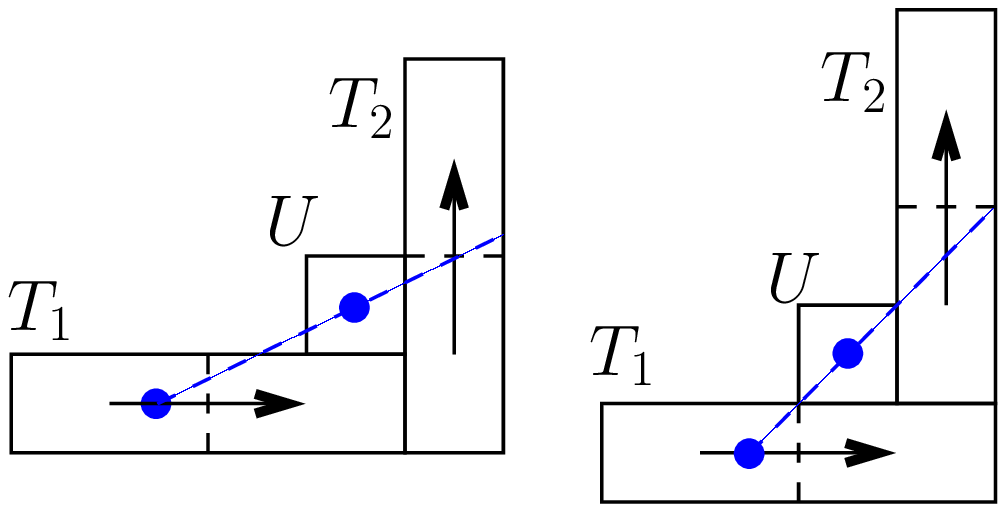}
\hspace{2cm}
\includegraphics[width=6cm]{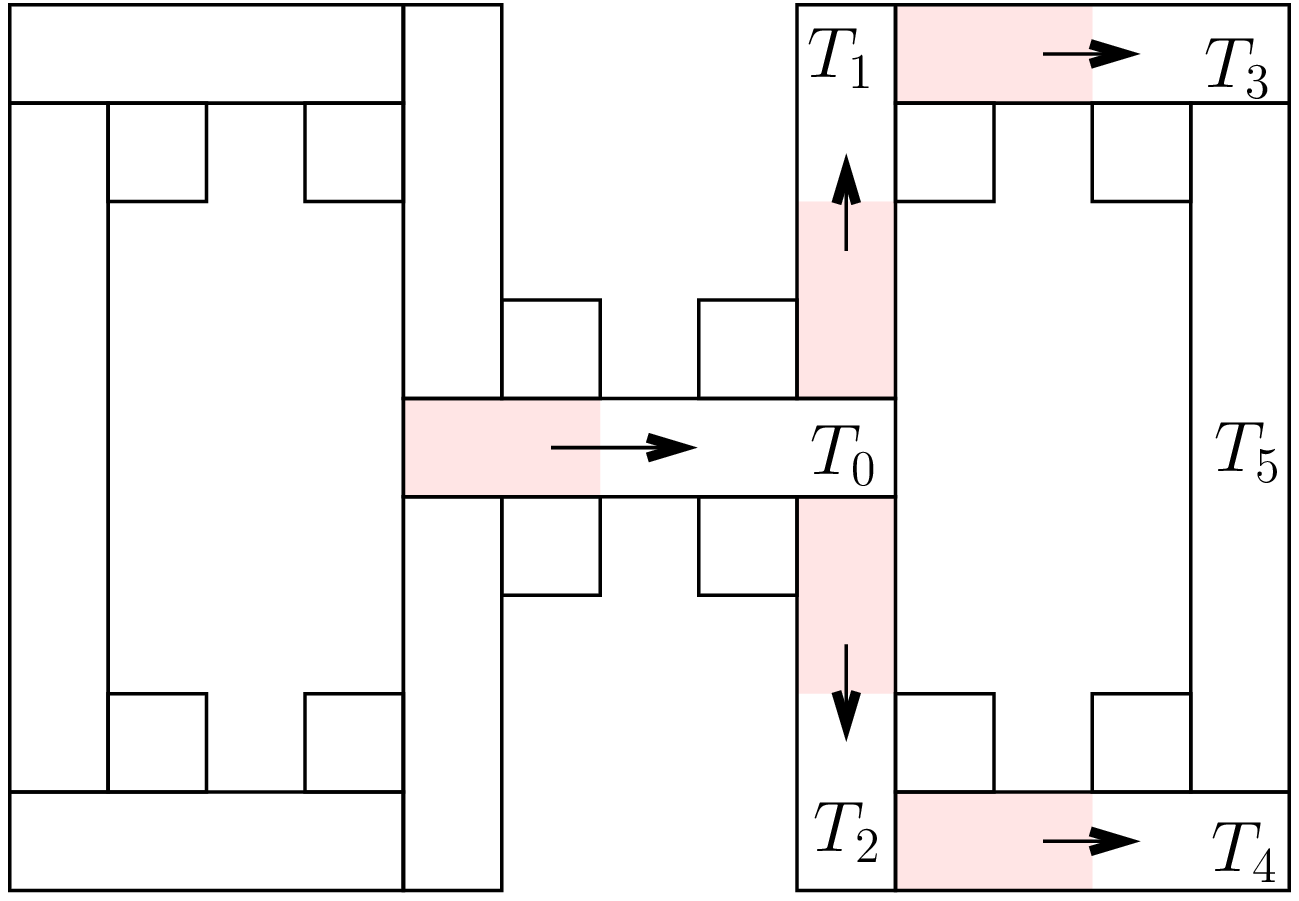}
\end{center}
\caption{Left: The two possible configurations of L-joints (up to rotations and reflections).
Right: A (partial) partition that does not permit an embedding.}
\label{fig:L-joint}
\label{fig:counter}
\end{figure}

\paragraph{Thin rectangles and L-joints}
We call an integral rectangle \emph{thin} if one of its side has length $1$, and the 
other side has length at least $4$.
The \emph{length} $\ell$ of a thin rectangle is the length of its longer side. 
Clearly, any projection has to map the vertex of a thin rectangle 
to one of the $\ell$ pixel centers.
Three intersecting rectangles are called an \emph{L-joint} if two of the rectangles
are thin and their union is L-shaped (up to rotation and reflection), and the third
rectangle is the pixel in the bulge of the L-shape; see Figure~\ref{fig:L-joint} (left).
If the L-joint belongs to a partition, it induces a $2$-simplex in $D(\partition)$.
An \emph{L-path} is a sequence of thin rectangles such that two consecutive elements
form an L-joint.
Such an L-path can be cyclic, in which case we call it an \emph{L-cycle}.
It is convenient to think of the thin rectangles in an L-joint being directed consistently
towards one end of the L-shape. If the direction is fixed, it makes sense to talk
about the \emph{front rectangle} and \emph{back rectangle} of the L-joint, 
about the \emph{front half} and \emph{back half} of each thin rectangle, and about 
the \emph{front pixel} and the \emph{back pixel} 
which are the extremal pixels in a thin rectangle.

\begin{lemma}[L-joint lemma]\label{lem:L-joint}
Let $(T_1,T_2,U)$ be a directed L-joint, with $T_1$ being the back rectangle, $T_2$
being the front rectangle, and $U$ being the bulge pixel. Let $\pi$ be a projection
that preserves the orientation of the $2$-simplex induced by the L-joint.
Then, if $\pi$ maps the vertex of $T_1$ to to its back half, it also maps the vertex of $T_2$
to its back half.
\end{lemma}
\begin{proof}
Figure~\ref{fig:L-joint} (left) displays the two possible connections of $T_1$ and $T_2$,
if both are of length $4$. 
Recall that $\pi$ is assumed to be half-integral and, therefore, has only $4$ choices to
place the vertex for both $T_1$ and $T_2$.
We can see that, if $\pi$ preserves the orientation
and maps the vertex of $T_1$ into its back half 
it must map the vertex of $T_2$ below the dashed line, so this vertex
must go into the back half in
both cases. The situation does not change if we extend $T_1$ to the left or $T_{2}$
to the top.
\end{proof}

We can use L-joints to show that dual complexes cannot always be embedded:
Consider the (partial) rectangular partition in Figure~\ref{fig:counter} (right).
Let $\pi$ be any projection of $D(\partition)$. 
Assume for a contradiction that $\pi$ preserves the orientation
of all $2$-simplices induced by the partial partition. 
Assume w.l.o.g.\ that the vertex of $T_0$ is projected
to the left half of $T_0$ (otherwise, the symmetric argument applies).
Direct the thin rectangles according to the arrows in the figure.
By the L-joint lemma, the vertices of $T_1$ and $T_2$ must be placed in their
back halves, and by repeating the argument, the same holds for the rectangles $T_3$ and $T_4$.
Using the L-joint Lemma again, the vertex of $T_5$ must be placed in the upper half of $T_5$
(caused by the L-joint with $T_3$) and also in the lower half of $T_5$ (caused by the L-joint
with $T_4$), a contradiction.
Filling out the partial partition with pixels, we obtain a full partition $\partition$ such that
no projection preserves the orientation of all triangles. The Embedding Theorem (Theorem~\ref{thm:embedding}) asserts
that there is no embedding.

\paragraph{The reduction}
We define the decision problem \texttt{Faithful\_Embed} as follows:
For a partition
$\partition=(R_0,\ldots,R_m)$ of $[0,n]\times [0,n]$, is there an embedding of $D(\partition)$?
Our goal is to show NP-completeness of this problem. 

First of all, the problem is clearly in NP: Given a partition $\partition$, 
we can compute $D(\partition)$ and check whether a specific projection causes an orientation switch 
for any triangle in polynomial time.

For the reduction, 
we define the \texttt{grid3sat} problem in the same way as in~\cite{godau-difficulty}:
\begin{quotation}
\texttt{grid3sat}: 
Given some $N\times N$ grid (with $N$ linear in $n$) in which some grid points are labeled
as clauses and some as variables. Variables are connected with clauses by vertex-disjoint
paths on the grid. A sign is associated to every such path indicating whether the corresponding
variable is negated in the clause or not. Every clause is incident to exactly three paths.
Is the formula described in this way satisfiable or not? 
\end{quotation}

See Figure~\ref{fig:grid3sat} (left) for an example.
A consequence of the formulation is that every variable appears in at most four clauses.
Therefore, the problem is a variant of the \emph{planar 3,4-SAT} problem
(see~\cite{godau-difficulty})
with the restriction that all paths follow grid edges. 

\begin{figure}
\centering
\includegraphics[width=4.2cm]{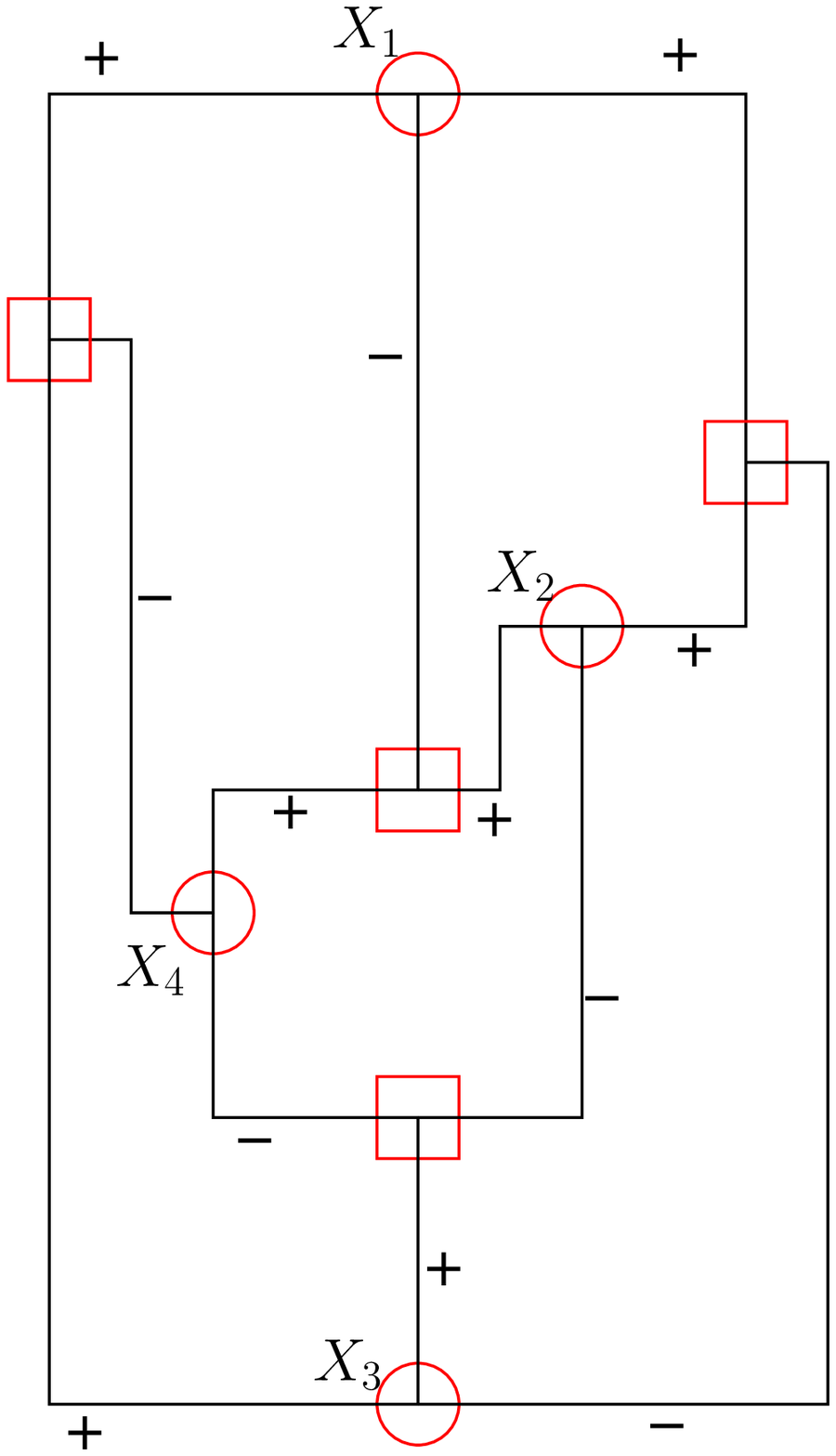}
\hspace{1cm}
\includegraphics[width=7.8cm]{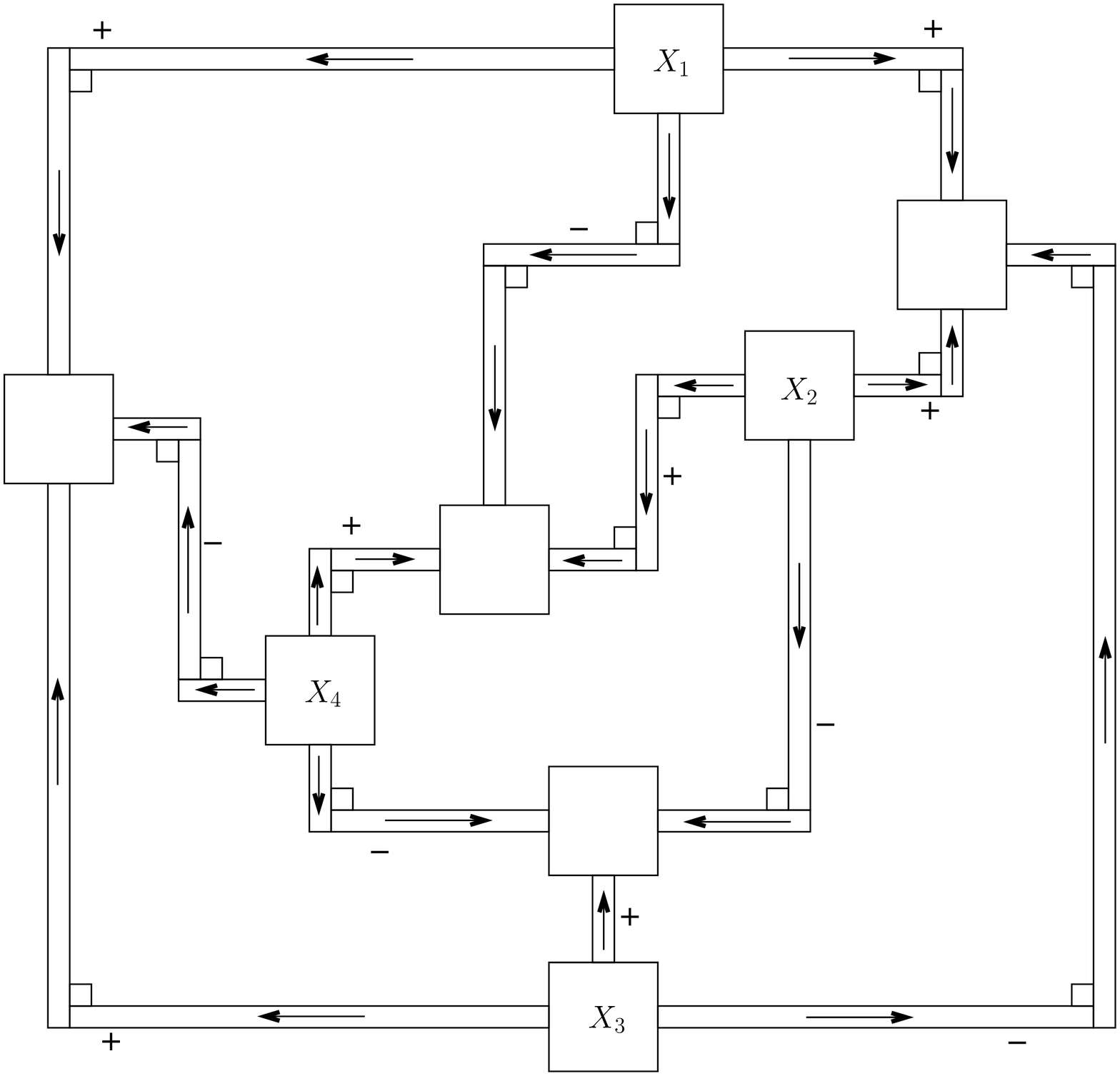}
\caption{Left: A graph presenting the formula $(X_1\vee X_2\vee \bar{X}_3)\wedge(\bar{X}_1\vee X_2\vee X_4)\wedge(X_1\vee\bar{X}_3\vee\bar{X}_4)\wedge(\bar{X}_2\vee X_3\vee\bar{X}_4)$.
Right: A corresponding representation as a partial integral rectangular partition, with
L-path connections.}
\label{fig:grid3sat}
\end{figure}

As stated in~\cite{godau-difficulty}, \texttt{grid3sat} is NP-complete.
Consider an instance $I$ of \texttt{grid3sat},
which is a graph on the integer grid representing a formula $F$ (in 3-CNF form).
Based on this graph, we will construct a rectangular partition $\partition$ such that the existence
of a satisfying assignment for $F$ is equivalent to the existence of an
embedding of $D(\partition)$. 

We barycentrically refine (see Section~\ref{sec:definitions}) the grid
given by $I$ a constant number of times 
(in fact $5$ barycentric refinements suffice) and construct our partition inside
this base grid. For that, we first replace vertices representing variables or clauses
by boxes of sufficient size (which we fill later), and we replace paths connecting
variables and clauses by disjoint L-paths. 
Each path inherits the sign of the corresponding path in $I$.
We think of L-paths being
directed from variables to clauses. See Figure~\ref{fig:grid3sat} (right).

\begin{figure}
\centering
\includegraphics[width=7cm]{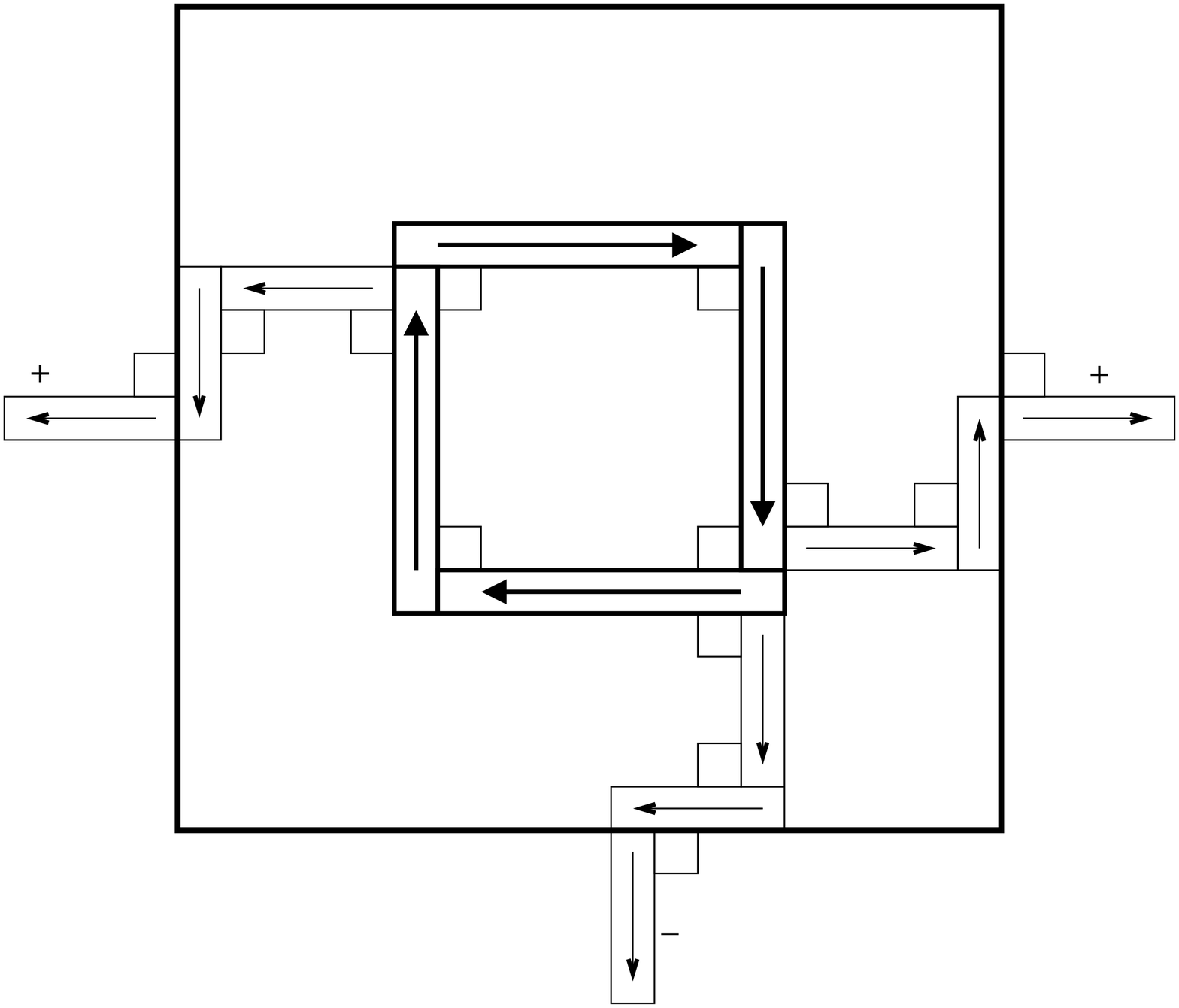}
\hspace{1cm}
\includegraphics[width=7cm]{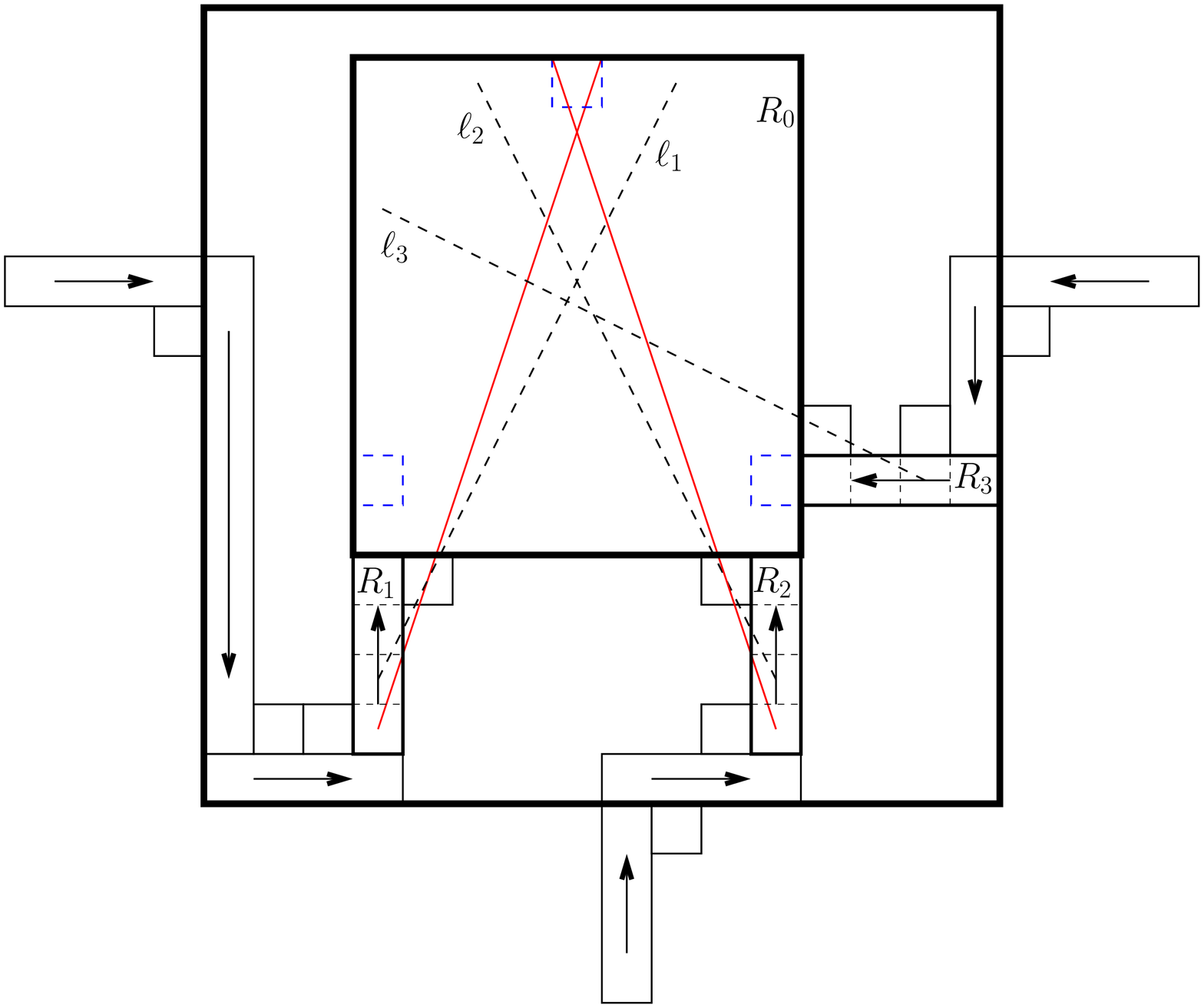}
\caption{Left: The variable gadget for $X_1$ from Figure~\ref{fig:grid3sat}.
Right: The clause gadget for $\bar{X}_2\vee X_3\vee\bar{X}_4$ from Figure~\ref{fig:grid3sat}.
Note that $R_1$, $R_2$, $R_3$ are precisely of length 4; the dashed lines and the red lines
are drawn for illustration purposes in Lemma~\ref{lem:clause_lemma}.
}
\label{fig:gadgets}
\end{figure}

Inside each box representing a variable $X_i$, we place a \textbf{variable gadget} 
(Figure~\ref{fig:gadgets} (left)).
In its center, it has an L-cycle with 4 thin rectangles of length $8$,
which we direct in clockwise direction.
We call this cycle the \emph{variable cycle for $X_i$}.
Each L-path to a clause box that starts at a boundary side of the gadget is extended
into the interior of the box and forms an L-joint with 
the corresponding side of the variable cycle.
If the path is associated with a plus sign (that is,
the variable appears in non-negated form in the corresponding clause), the path
is connected to the front pixel. If it is associated with a minus sign, it is connected
to the back pixel. 

Inside each box representing a clause, we place a \textbf{clause gadget}
(Figure~\ref{fig:gadgets} (right)).
The three L-paths arriving at the boundary of the gadget are extended to connect
to the \emph{clause rectangle} (rectangle $R_0$ in Figure~\ref{fig:gadgets} (right)), 
with the last rectangle being exactly $4$ base squares long.
If the L-paths arrive at different sides than displayed in 
Figure~\ref{fig:gadgets}, we can just rotate the gadget by a multiple of $90^\circ$.
The gadget is designed to satisfy the following two properties:

\begin{lemma}[Clause Lemma]\label{lem:clause_lemma}
Let $v_0,v_1,v_2,v_3$ be the vertices of $R_0,R_1,R_2,R_3$,
respectively, 
as in Figure~\ref{fig:gadgets} (right), and let $\pi$ be
a partial projection that maps $v_i$ to a point $p_i$ for $1\leq i\leq 3$.
\begin{enumerate}
\item If $p_1$, $p_2$, and $p_3$ are in the back half
of $R_1$, $R_2$, and $R_3$, respectively, then $\pi$ cannot be completed to an embedding.
\item If at least one of the $p_i$ is in the front pixel of $R_i$, for $1\leq i\leq 3$, 
then there exists an extension of $\pi$ to $v_0$
such that $\pi$ preserves the orientation of all triangles incident to $v_0$.
\end{enumerate}
\end{lemma}
\begin{proof}
Consider Figure~\ref{fig:gadgets} (right);
for the first part, let $p_0$ be the projection of $v_0$ in a completion of $\pi$. 
If $\pi$ is an embedding, $p_0$ is left of $\ell_1$,
right of $\ell_2$ and below $\ell_3$, which is impossible.
For the second part, if $p_1$ or $p_2$ are in the front pixel,
we can set $p_0$ to be the center of the rightmost or leftmost dashed base square in 
Figure~\ref{fig:gadgets} (right), respectively. If $p_3$ is in the front pixel, we can
choose the center of the topmost blue base square for $p_0$. From the red line
originating in $R_1$ in Figure~\ref{fig:gadgets} (right), it can be seen that the triangle
induced by $R_0$, $R_1$, and the pixel adjacent to both has the correct orientation.
The same is true for the triangle spanned by $R_0$, $R_2$ and the pixel
adjacent to both.
\end{proof}

To summarize, we have constructed a partial integral rectangular partition, consisting of
variable cycles, clause rectangles, and L-paths connecting them. We finally 
complete the partition by filling the empty spots with pixels. Let $\partition$ denote
this partition. We can construct $\partition$ in polynomial time in $n$. The following lemma
is sufficient to prove NP-hardness of \texttt{Faithful\_Embed}.

\begin{lemma}\label{lem:sat_emb_equi}
There exists an embedding of $D(\partition)$ if and only if $F$ has a satisfying assignment.
\end{lemma}
\begin{proof}
``$\Rightarrow$'': Let $\pi$ be an embedding. We define an assignment as follows: For
every variable $X_i$, consider its variable cycle (which is clockwisely arranged).
Note that the L-joint Lemma implies that $\pi$ either projects all vertices
of the variable cycle to the front half, or it projects all vertices to the back half
of the corresponding rectangles. 
We set $X_i$ to \True\ if $\pi$ projects to the front half
and to \False\ otherwise. We show that this assignment satisfies $F$:

Let $C=L_i\vee L_j\vee L_k$ be a clause of $F$ where $L_i=X_i$ or $L_i=\bar{X}_i$,
and consider the corresponding clause gadget. Let $R_i$, $R_j$, $R_k$ the last rectangles
of the L-paths connecting the variable gadgets of $X_i$, $X_j$, and $X_k$, respectively,
to the clause rectangle. Since $\pi$ is an embedding, at least one of the vertices of $R_i$, $R_j$,
or $R_k$ is placed in the corresponding front half by the Clause Lemma.
Assume w.l.o.g.\ that the vertex of $R_i$ is placed in the front half. By iteratively applying
the L-joint Lemma,
all vertices in the L-path from $X_i$ to $C$ must therefore be placed in the front half.
If $L_i=X_i$, the path was associated with a plus sign, and therefore, the L-path is connected
with the variable cycle of $X_i$ at a front pixel 
(as for instance the left and right paths originating from
the variable cycle in Figure~\ref{fig:gadgets} (left)). Since $\pi$ preserves orientations, 
it follows again from the L-joint Lemma that the
vertices in the variable cycle must also have been placed in the front half. 
Hence, our assignment
sets $X_i$ to $\True$ and therefore, the clause is satisfied. If $L_i=\bar{X}_i$, 
the argument is analogous.

``$\Leftarrow$'': We construct a projection based on a satisfying 
assignment. We start with the variable gadgets: If $X_i=1$, we place the vertices
in its variable cycle at the front pixel, otherwise at the back pixel.
For every L-path that leaves a variable cycle, we distinguish two cases: If
the path has a plus sign and $X_i=1$, or if the path has a minus sign and $X_i=0$, we place
all vertices in the path in the front pixel; otherwise, we place all vertices in the back
pixel. It is straight-forward to verify that this projection preserves the orientation
for each $2$-simplex induced by L-joints.
Finally, let $C=L_i\vee L_j\vee L_k$ be a clause as above. 
Consider the corresponding clause
gadget and let $R_i$, $R_j$, $R_k$ denote the connecting
(thin) L-path rectangles. By our construction, the vertex of $R_i$ is in the front pixel
if and only if $L_i=1$, same for $j$ and $k$. 
Since we have a satisfying assignment, at least one literal is
satisfied, so one of the vertices is in the front pixel. By the Clause Lemma,
we can therefore place the vertex of the clause rectangle such that all triangles
preserve their orientation.

Our projection preserves all orientations of the constructed partial partition.
It is not difficult to see that the additional $2$-simplices caused by filling the empty spots
with base squares preserve their orientation as well (in fact, any projection preserves
their orientation). Therefore, the constructed projection is an embedding by the 
Embedding Theorem.
\end{proof}

In conclusion, we can summarize

\begin{theorem}
\texttt{Faithful\_Embed} is NP-complete.
\end{theorem}

A slight variation of our proof
shows that the problem remains NP-hard when allowing projections that are not half-integral
(see Appendix~\ref{app:non-half-integral}).
However, it is unclear whether this variant of the problem also remains in NP.

\section{Center projections and $\beta$-balancing}
\label{sec:balancing}

We turn back to arbitrary dimensions.
We define the \emph{center} of a box $[a_1,b_1]\times\ldots\times [a_d,b_d]$ 
to be the point
$((a_1+b_1)/2,\ldots,(a_d+b_d)/2)$. The \emph{center projection} is the projection of $D(\partition)$
that maps each vertex to the center of the corresponding box. Clearly,
this projection is half-integral and faithful. We call $\partition$ \emph{center-embeddable} if the center projection
is an embedding. 

We define the \emph{balance} of
a set of boxes $\{R_0,\ldots,R_k\}$
to be the length of the longest side among all boxes $R_0,\ldots,R_k$ divided by the length of the shortest side
among all boxes $R_0,\ldots,R_k$.
We define the balance of a $k$-simplex $\sigma$ of $D(\partition)$ to be the balance of the set of dual boxes.
The \emph{aspect ratio} of a box $R$ is the balance of $\{R\}$.
Obviously $R$ has aspect ratio $1$ if and only if $R$ is a square box.
We call $D(\partition)$ \emph{$\beta$-balanced} for some $\beta\geq 1$, if each simplex
of $D(\partition)$ has balance at most $\beta$.\footnote{This is equivalent to require
that any edge in $D(\partition)$ has balance at most $\beta$, which is the definition
given in the introduction.}
Informally, 
$\beta$-balanced partitions have the property that 
boxes are not too skinny, but also
neighboring boxes do not differ
too much in side lengths. 

The simplest sufficient condition for a partition to be center-embeddable is that all boxes
are square boxes of same size. We investigate several generalizations of this trivial criterion.

\paragraph{Planar results}
It is not hard to prove that a partition in $\R^2$ consisting 
only of squares is always center-embeddable, 
because the edge connecting the centers of two squares does not leave the union of the squares.
However, if we only bound the aspect ratio of any rectangle by $1+\eps$ (with $\eps>0$),
this property does not hold in general, and we can construct a counterexample
for center-embeddability.
For $\beta$-balanced partitions, we prove a tight bound for center-embeddability. 

\begin{figure}
\centering
\includegraphics[width=6cm]{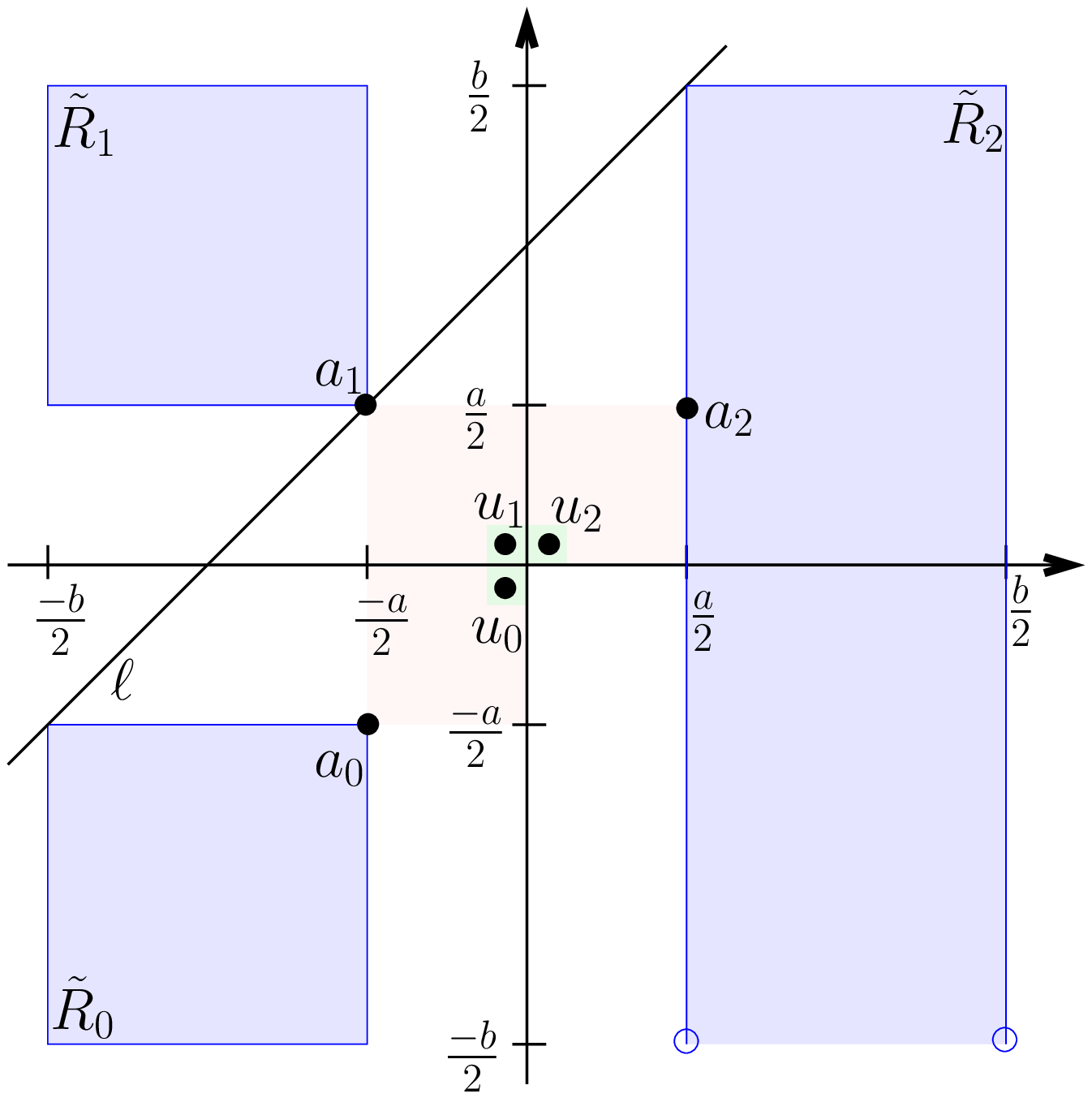}
\hspace{1cm}
\includegraphics[width=7.5cm]{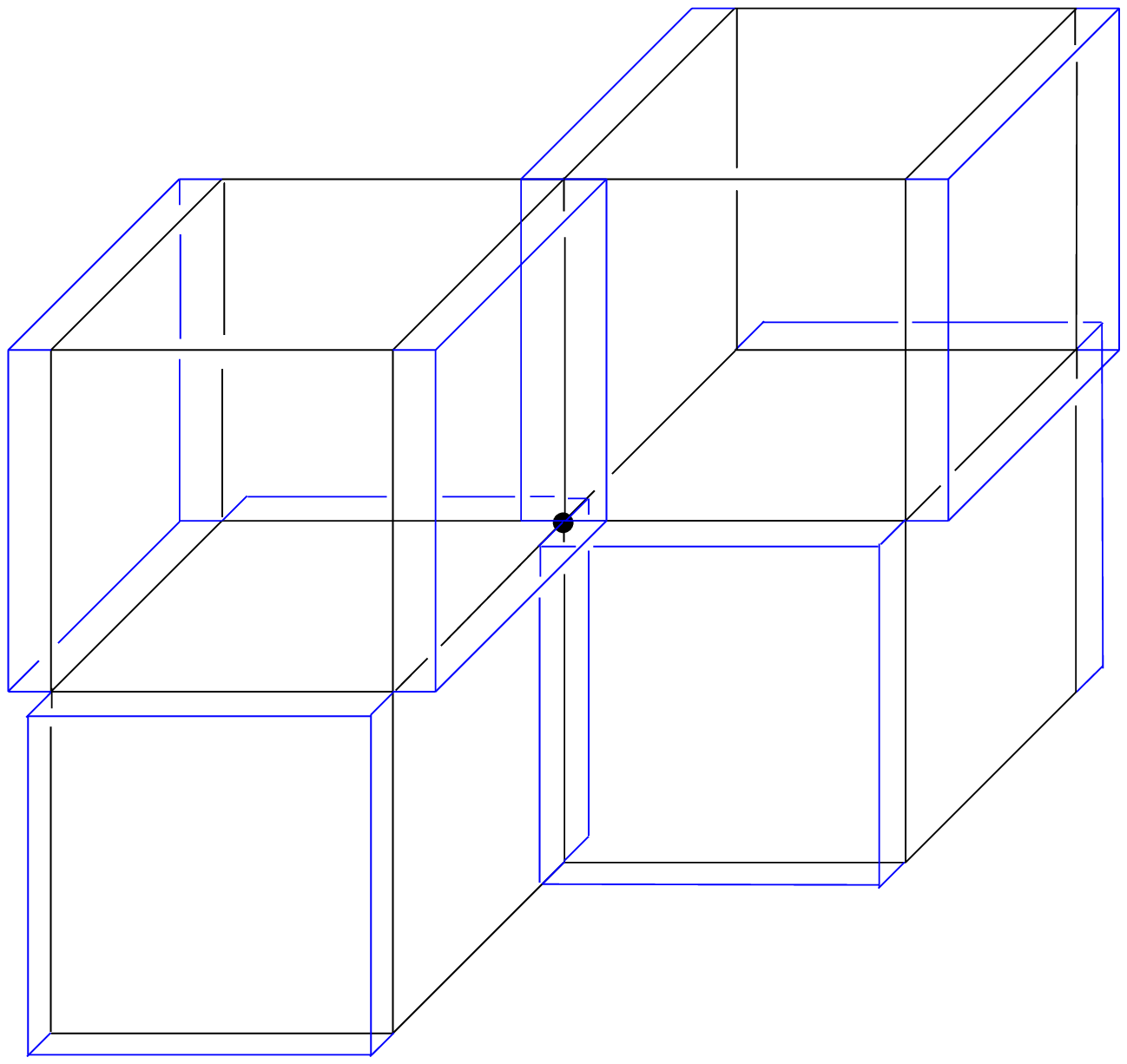}
\caption{Left: Illustrations of the proof of Theorem~\ref{thm:planar_balancing}.
The drawing assumes that $\beta=3$. In this case, we can find a line that touches
$C_0$, $C_1$, and $C_2$, leading to a counterexample for center-embeddability.
Right: Illustration of the construction in the proof of Theorem~\ref{thm:beta_balanced_high}.}
\label{fig:worst_case_3d}
\label{fig:beta_balanced_plane}
\end{figure}

\begin{theorem}[Planar Balancing Theorem]\label{thm:planar_balancing}
A $\beta$-balanced partition of $\R^2$ is center-embeddable for $\beta<3$, and
there exists a $3$-balanced partition that is not center-embeddable.
\end{theorem}
\begin{proof}
Consider a $\beta$-balanced partition. For the first part, it is enough to prove that the center projection
preserves the orientation of each $2$-simplex for $\beta<3$ according to the Embedding Theorem.
So, let $\sigma$ be a $2$-simplex, caused by the intersection of three rectangles $R_0$, $R_1$,
$R_2$. W.l.o.g., we assume that the intersection point is the origin, that $R_0$ is completely
contained in the lower-left quadrant, $R_1$ is contained in the upper-left quadrant,
and $R_2$ is contained in the halfplane $x\geq 0$; all other cases are obtained by 
suitable rotations and reflections. Let $a$ denote the length of the shortest side
among the $R_i$, and let $b$ denote the length of the longest side; we have $\frac{b}{a}\leq\beta$.
We let $C_i$ denote the set of possible positions of the center of $R_i$. 
By the imposed constraints, it is straight-forward to verify that $C_i$
are rectangles as illustrated in Figure~\ref{fig:beta_balanced_plane} (left).
Note that the lower side of $C_2$ does not belong to the set of possible centers,
while the other boundary sides do.

Recall that the orientation of $\sigma$ is preserved if
$\orient(p_0,p_1,p_2)=\orient(u_0,u_1,u_2)$, where the $u_i$ define the seed configuration
of $\sigma$; in our case, $u_0=(-1/2,-1/2)$, $u_1=(-1/2,1/2)$, and $u_2=(1/2,1/2)$.
We scale each $u_i$ by $a$, obtaining the points $a_0$, $a_1$, $a_2$; 
obviously, $\orient(u_0,u_1,u_2)=\orient(a_0,a_1,a_2)$.
Assume that the orientation of $\sigma$ is not preserved.
Then, if we continuously move the points $a_i$ to the points $p_i$ along the
connecting line segment (which lies in $C_i$), there must be a time
where the three points become collinear. 
It follows that there must be a line in $\R^2$ 
which intersects each $C_i$. 
Such a line, however, exists if and only if $b\geq 3a$,
as can be seen by elementary geometric arguments.
It follows that the orientation of $\sigma$ is preserved if $\beta<3$.

The construction also yields the counterexample for $\beta=3$: Consider the rectangles
intersecting in the origin with centers
$p_0=(-3/2,-1/2)$, $p_1=(-1/2,1/2)$, $p_2=(1/2,3/2)$. We can fill the bounding box
by pixels to obtain a $3$-balanced partition. The $2$-simplex dual to the origin
is projected to a line segment because $p_0$, $p_1$, $p_2$ are collinear. Therefore,
the partition is not center-embeddable.
\end{proof}

The counterexample constructed in the proof of Theorem~\ref{thm:planar_balancing}
yields a non-generic partition. If we assume genericity, the set $C_2$ defined in the proof
loses the upper boundary side, and a line intersecting each $C_i$ can only be found if $b>3a$.
It follows that generic $3$-balanced partitions are center-embeddable. 
However, generic $\beta$-balanced partitions that are not center-embeddable 
can be constructed for any $\beta>3$.

\paragraph{Higher dimensions} Restricting to square boxes (that is, aspect ratio $1$) 
does not guarantee
center-embeddability for any $d>2$: In $\R^3$, the triangle spanned by the centers
of three intersecting cubes might leave the union of the three cubes, and a counterexample
can be constructed; see~\cite{bek-computing} and~\cite{ek-freudenthal}.
We show that also restricting to $\beta$-balanced partitions does not guarantee embeddability,
unless $\beta=1$:

\begin{theorem}\label{thm:beta_balanced_high}
For any $\beta>1$, there exists a $\beta$-balanced partition in $\R^d$ that is not center-embeddable.
\end{theorem}
\begin{proof}
We restrict to $d=3$ for the simplicity of presentation
although our construction generalizes to any dimension:
For simplicity, we talk about \emph{cubes} and \emph{voxels}
instead of square boxes and unit boxes throughout the proof.
Choose $b\in\Z$ such that $\beta\geq \frac{b+2}{b}>1$. 
Consider four cubes of length $b$ having the origin as a corner,
such that the centers lie in a common plane 
(the black boxes in Figure~\ref{fig:worst_case_3d} (right)).
Then, extend the cubes by a layer of voxels at two opposite faces, such that 
all voxels around the origin are covered by the four boxes
(the blue boxes in Figure~\ref{fig:worst_case_3d} (right)). 
This ensures that the $3$-simplex induced by the four boxes belongs to the dual complex
of the partition.
The extension does not change the centers, consequently,
the center projection cannot preserve the orientation of that simplex.
Moreover, all boxes have side length $b$ and $b+2$. It remains to show that
the arrangement of the four boxes can be completed to a partition without
destroying $\beta$-balancing. We give details of this step
in Appendix~\ref{app:puzzle}.
\end{proof}

\paragraph{The cubical case} We finally consider $\beta$-balanced partitions
in $\R^d$ where all boxes are square boxes. 
We call such partitions \emph{cubical}.
We first construct
counterexamples in arbitrary dimension to show:

\begin{theorem}\label{thm:counter_beta_balanced_dd}
For $d\geq 3$ and $\beta>\frac{d}{d-2}$, 
there exists a $\beta$-balanced cubical partition in $\R^d$ that is not
center-embeddable.
\end{theorem}
\begin{proof}
For $d\geq 3$, $a,b\in\Z$ positive, and $\delta\in\{0,1\}$ consider the $(d+1)\times (d+1)$-matrix
\[M^{(\delta)}_d:=\left(\begin{array}{cccccc}
1 & -\frac{a}{2} &  & \ldots &  & -\frac{a}{2} \\
1 & \frac{b}{2} & -\frac{b}{2} & \ldots & \ldots & -\frac{b}{2}\\
1 & -\frac{b}{2}+\delta & \frac{b}{2} & -\frac{b}{2} & \ldots & -\frac{b}{2}\\
\vdots & \vdots & \ddots & \ddots & \ddots & \vdots\\
1 & -\frac{b}{2}+\delta & \ldots & -\frac{b}{2}+\delta & \frac{b}{2} & -\frac{b}{2}\\
1 & -\frac{b}{2}+\delta & \ldots & \ldots & -\frac{b}{2}+\delta & \frac{b}{2}
\end{array}\right).\]
It can be shown by elementary row operations (Appendix~\ref{app:determinant}) that
\[\det M^{(0)}_d=\frac{1}{2}b^{d-1}(da-(d-2)b).\]
In particular, note that the determinant is positive 
if $\frac{b}{a}<\frac{d}{d-2}$, vanishes if $\frac{b}{a}=\frac{d}{d-2}$,
and is negative if $\frac{b}{a}>\frac{d}{d-2}$.

Let $a$ and $b$ be such that 
$\beta>\frac{b}{a}> \frac{d}{d-2}$, so that
$\det M^{(0)}_d<0$. If we choose $a$ and $b$
large enough, it follows that $\det M^{(1)}_d<0$
as well, just because $M^{(1)}_d$ constitutes
a small perturbation of $M^{(0)}_d$ for large $a$ and $b$. 
The rows of $M^{(1)}_d$ 
define $d+1$ points $p_0,\ldots,p_d$ in $\R^d$ (ignoring the first column).
We define the square box $R_i$ with center $p_i$ and side length $a$
for $R_0$, and length $b$ for $R_1,\ldots,R_d$. By this choice, it can be seen
easily that all $2^d$ unit boxes around the origin are covered by
the $R_i$, so $R_0,\ldots,R_d$ form a $d$-simplex $\sigma$ in the dual complex.
Moreover, $u_0:=(-1/2,\ldots,-1/2)\in R_0$, $u_1:=(1/2,-1/2,\ldots,-1/2)\in R_1$,
$u_2:=(1/2,1/2,-1/2,\ldots,-1/2)\in R_2$, and so on, thus $u_0,\ldots,u_d$
form a seed configuration and it is easily seen to that $\orient(u_0,\ldots,u_d)>0$.
Because $\orient(p_0,\ldots,p_d)<0$ by construction, the center
projection does not preserve the orientation of $\sigma$.

The last step of the proof is to show that we can complete the initial
configuration $R_0,\ldots,R_d$ to a complete $\beta$-balanced
cubical partition in $\R^d$. The details of this step are 
skipped for brevity. See Appendix~\ref{app:puzzle} for details.
\end{proof}

We are able to show that the constructed counterexample is the worst case
in $\R^3$. Precisely, we state:

\begin{theorem}\label{thm:cubical_partition_bound_3d}
Cubical $\beta$-balanced partitions in $\R^3$ 
are center-embeddable for $\beta<3$.
\end{theorem}
\begin{proof}
For simplicity, we restrict to generic partitions in the proof; the 
general case works with the same methodology, but requires closer
investigation of the distortion defined for the dual complex construction.
As in Theorem~\ref{thm:beta_balanced_high}, 
we use the terms cube and voxel instead of square box and hyper-pixel.

Consider $4$ cubes $R_0,\ldots,R_3$ intersecting in a point $p$.
Let $a$ be the length of the shortest and $b$ be the length of the longest
side among the $R_0,\ldots,R_3$. 
Assume w.l.o.g.\ that $p=(0,0,0)$. By genericity, the union of the $R_i$ covers all
$8$ voxels adjacent to $p$, and each cube covers at least one voxel.
Represent these voxels by the 3-bit numbers $0,\ldots,7$ where the 
$1$st/$2$nd/$3$rd bit of voxel $j$ is set to $1$ if and only if the 
$x$-/$y$-/$z$-coordinate of the voxel is positive. 
Assign to each $R_i$ a subset of
$\{0,\ldots,7\}$, denoting the voxels that it occupies.
Only three cases are possible, namely that $R_i$
covers a single voxel, that $R_i$ covers two voxels which are face adjacent,
or four voxels which lie in a common halfspace.
We consider $C_i$, the set of possible centers for $R_i$.
If $R_i$ covers only one voxel
adjacent to the origin, say voxel $0$, then $C_i$
is the line segment connecting $(-a/2,-a/2,-a/2)$
and $(-b/2,-b/2,-b/2)$. If $R_i$ covers two voxels which are face adjacent,
say voxels $0$ and $1$, $C_i$ is the trapezoid
spanned by the four points $(\pm a/2,-a/2,-a/2)$, $(\pm b/2,-b/2,-b/2)$; 
more precisely, the center can not
lie on either of the two non-parallel lines because this would prevent
the cube to span over both voxels, so $C_i$ is the trapezoid with those
two sides excluded.
Finally, if $R_i$ covers four voxels in a common hyperplane, say
voxels $4,\ldots,7$, $C_i$ is the polytope spanned by
the $8$ points $(\pm a/2, \pm a/2,a/2)$, $(\pm b/2, \pm b/2, b/2)$;
the polytope has 6 faces, two
of them being parallel.
More precisely, the center can not lie on a boundary face except for the
two parallel sides, so $C_i$ is that polytope with 
the other $4$ faces excluded.

Assume for a contradiction that the center projection does not preserve
the orientation of the simplex spanned by $R_0,\ldots,R_3$.
Scaling the seed configuration of the simplex by a factor of $a$, 
we obtain a simplex spanned by $a_0,\ldots,a_3$ with $a_i$ in the closure
of $C_i$ with the same orientation as the seed configuration. Because
this orientation is not preserved, there are points $p_0,\ldots,p_3$
with $p_i\in C_i$ such that the orientation is different. It follows
that there are points $p'_0,\ldots,p'_3$ with
$p'_i\in C_i$ whose orientation is zero. In other words, there exists
a plane which intersects $C_0,\ldots,C_3$.
However, $C_0,\ldots,C_3$ are completely determined
by the way of how the boxes $R_i$ decompose the voxels $\{0,\ldots,7\}$
adjacent to the origin, and there are only two configurations possible,
up to rotations and reflections:
the \emph{regular configuration}
\[\{0\}, \{1\}, \{2,3\}, \{4,5,6,7\},\]
and the \emph{singular configuration}
\[\{0,1\}, \{2,3\} ,\{4,6\}, \{5,7\}.\]
Consequently, there must be a plane intersecting the two line segments, the
trapezoid and the polytope defined by the regular configuration,
or a plane intersecting the four trapezoids defined 
by the singular configuration. It can be verified, however, that such a plane
does not exist by formulating the statement in terms of a quantified
system of inequalities and using a quantifier elimination program~%
\cite{ch-partial}\footnote{We used
qepcad: \url{http://www.usna.edu/cs/~qepcad/B/QEPCAD.html}} (see Appendix~\ref{app:qe}). 
Hence, the center-projection
preserves the orientation of every simplex and is an embedding.
\end{proof}

The proof idea can in principle be extended to higher dimensions. However,
there are two problems: First, the number of configurations to check
increases; for instance, there are already 3 configurations to check in $\R^4$.
Second, and more seriously, the complexity of the quantifier elimination
increases dramatically for higher dimensions; in fact, we were not able
to verify that $\beta$-balanced cubical partitions in $\R^4$ 
are center-embeddable for $\beta<\frac{d}{d-2}=2$.

\section{Conclusion}
\label{sec:conclusion}

We have proven positive and negative results about partitions and their duals 
in this paper; it is NP-hard to
decide whether an embedding exists already in the plane; if we project
vertices to centers of rectangles, there are some simple balancing
conditions that ensure embeddability in the plane, namely, if
all rectangles are squares or if the partition is $(3-\eps)$-balanced
with $\eps>0$. In higher dimensions, however, we need to combine
both properties, that means, consider cubical $(3-\eps)$-balanced partitions
to get a guarantee that partitions can be embedded.
We are posing the question whether another simple condition could
guarantee center-embeddability as well.

Quad-tree decompositions and their higher-dimensional analogues have
the nice property that they can be balanced with a greedy strategy,
without granulating the partition too much. 
It seems unclear, however, how a planar partition can be
turned into a $3$-balanced partition efficiently by subdividing
rectangles.

This work has considered
hyper-rectangular partitions in general~-- can better embeddability
results be derived for certain restrictions,
in particular for partitions that arise from a sequence of 
hyperplane cuts (sometimes also called guillotine constructions)?
We expect the answer to be negative in the plane; however, the
situation might be different in higher dimensions.

On a technical side, a natural improvement would be to replace
the computer-assisted proof in Theorem~\ref{thm:cubical_partition_bound_3d}
with a geometric argument; this would hopefully extend into higher dimensions
and prove our conjecture that $\beta$-balanced partitions in $\R^d$
are center-embeddable for $\beta<\frac{d}{d-2}$.
A second question is about a variant of \texttt{Faithful\_embed}: 
while it is clear
that deciding the existence of an embedding remains hard in three
dimensions, it is not clear whether the restriction to cubical
partitions simplifies the problem.

\section*{Acknowledgments}
The author thanks Herbert Edelsbrunner for fruitful discussions,
David Eppstein for pointing out related work,
and the anonymous referees for their extremely valuable feedback.

\begin{appendix}

\section{The Embedding Theorem}
\label{app:embedding}

\begin{embtheorem}[Embedding Theorem]
Let $\partition$ be a partition.
Let $\pi$ be a projection of $D(\partition)$ that preserves
the orientation of at least one $d$-simplex.
Then, $\pi$ is an embedding if and only if 
it preserves the orientation of each $d$-simplex.
\end{embtheorem}

We give some more details about its proof which is based on~\cite{ek-freudenthal}.
A set is called \emph{contractible} if it is can be reduced to a single point
by a continuous deformation~\cite{hatcher}.

\begin{lemma}\label{lem:fractual_distortion}
Any intersection of $k$ distorted boxes of $\partition$ is either empty or contractible.
\end{lemma}
\begin{proof}
The proof given in the ``Fractual Distortion Lemma'' in~\cite{ek-freudenthal}
extends to arbitrary partitions (it is neither exploited that the boxes are cubes, nor
that they are arranged hierarchically).
\end{proof}

We now give a proof of the Embedding Theorem which works similar to
the proof of the ``Geometric Realization Theorem'' from~\cite{ek-freudenthal}.

Assume first that $\pi$ preserves the orientation of all $d$-simplices of $D(\partition)$.
Let $B\subset\R^d$ be the box that is partitioned by $\partition$. 
We attach two layers of unit boxes at the boundary of $B$, obtaining an extended box $\bar{B}$.
Note that this yields a partition $\bar{\partition}$, which is an extension of $\partition$.
We extend the projection $\pi$ of $\partition$ to $\bar{\pi}$ of $\bar{\partition}$ 
by mapping vertices in the new layers to the corresponding box centers.

We compactify $\R^d$ to the $d$-dimensional sphere $\S^d$ by adding a vertex at infinity.
Similarly, we compactify the dual complex $D(\bar{\partition})$ by adding a new vertex at infinity
and connecting it to every simplex of the boundary of $D(\bar{\partition})$.
Because of the Lemma~\ref{lem:fractual_distortion}, 
we can apply the Nerve Theorem~\cite{eh-computational} to $D(\bar{\partition})$, which states
that $D(\bar{\partition})$ triangulates a $d$-dimensional ball, and therefore, its compactification
triangulates $\S^d$.
Hence, the extended projection $\bar{\pi}$ defines a continuous mapping $g:\S^d\rightarrow\S^d$.
The \emph{degree of $g$} at a point $x$ not in any $(d-1)$-simplex is the number of
$d$-simplices that contain $g^{-1}(x)$, counting a $d$-simplex positive or negative
depending on the orientation of its image under $g$. This degree is $1$ in between the two
outside layers, because all boxes are of unit size and therefore, the orientation
of all $d$-simplices can be chosen positive. However, the degree of a map is a global
property that does not depend on the specific location of $x$. Hence, it is $1$ for any $x$.
 
We can assume that the orientation of all seed configurations is positive; because
all orientations are preserved by $\pi$,
the orientation of all $d$-simplices is positive. Because the degree of $g$ is $1$, it
follows that every point is in the interior of exactly one projected $d$-simplex. This implies injectivity, 
so $\pi$ is an embedding.

\begin{figure}[htb]
\centering
\includegraphics[width=7.5cm]{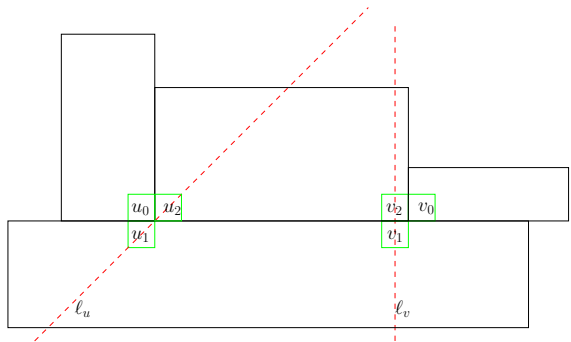}
\caption{Illustration for the fact that the orientations of $(u_0,u_1,u_2)$ and $(v_0,v_1,v_2)$
         are not the same; indeed, if we map $u_i$ to $v_i$ for $0\leq i\leq 2$, the line $\ell_u$
         is mapped to the line $\ell_v$, but the halfspace left of $\ell_u$ is mapped to
         the halfspace right of $\ell_v$. The same kind of argument can be applied
         in higher dimensions, replacing the dashed lines by hyperplanes and defining a notion
         of ``left'' and ``right'' by a common direction, for instance the vector from
         the center of $v_0$ to the center of $u_0$.}
\label{fig:orient_fig}
\end{figure}

For the opposite direction, assume that $\pi$ does not preserve the orientation 
of some $d$-simplex of $D(\partition)$. Because there exists at least one $d$-simplex
whose orientation is preserved, there exists some $(d-1)$-simplex $\tau$ in $D(\partition)$
incident to one simplex $\sigma_1$ whose orientation is preserved and to one 
simplex $\sigma_2$ whose orientation is not preserved.
It suffices to prove that $\sigma_1$ and $\sigma_2$ are mapped to the same side
of the supporting $(d-1)$-hyperplane of $\tau$.

To see that, let $u_0,\ldots,u_n$ and $v_0,\ldots,v_n$ denote the seed configurations of $\sigma_1$
and $\sigma_2$, respectively. W.l.o.g., we assume that the seed configurations are ordered such that
$u_1,\ldots,u_n$ and $v_1,\ldots,v_n$ are the pixels that belong to boxes in $\tau$, and
furthermore, for all $1\leq i\leq n$, $u_i$ and $v_i$ belong to the same box of the partition.
By these choices, we have that
$$O(u_0,\ldots,u_n)= -O(v_0,\ldots,v_n).$$
See Figure~\ref{fig:orient_fig} for an illustration in $\R^2$; 
the general case can be derived with an argument similar to~\cite[p.167]{ah-topologie}.
Now, let $w_1,\ldots,w_n$ denote the projections of the boxes of $\tau$, and let $u$ and $v$
denote the projections of the boxes corresponding to $u_0$ and $v_0$, respectively.
W.l.o.g., let $\sigma_1$ be spanned by $u$ and $w_1,\ldots,w_n$ and $\sigma_2$ be spanned
by $v$ and $w_1,\ldots,w_n$. Because the orientation of $\sigma_1$ is preserved, and the orientation
of $\sigma_2$ is not preserved, we have that
$$O(u,w_1,\ldots,w_2)=O(v,w_1,\ldots,w_2).$$
It follows that $u$ and $v$ are mapped to side of the hyperplane spanned by $w_1,\ldots,w_n$.

\section{Non-integral projections}\label{app:non-half-integral}
In this section, we reconsider the decision problem \texttt{Faithful\_Embed}
without the restriction that our projections are half-integral.
We show that the problem remains NP-hard in this relaxed setup.

The proof strategy follows the same steps as in Section~\ref{sec:NP}.
We require, however, slightly more skinny rectangles: We call a rectangle
\emph{thin} if one of its sides has length one, and the other has length at least $8$
(as opposed to length $4$ in Section~\ref{sec:NP}). The definitions of L-path, L-joint,
front pixel, front half etc.~directly carry over. Also, the L-joint lemma remains
true for arbitrary projections; see Figure~\ref{fig:L-joint-mod}.

\begin{figure}
\begin{center}
\includegraphics[width=6cm]{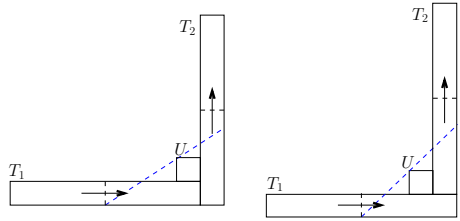}
\end{center}
\caption{The two possible configurations of L-joints (up to rotations and reflections).
If the vertex of $T_1$ is set in the back half, the pixel of $T_2$ must also be set
in the back half, regardless of where the vertex of $U$ is placed.}
\label{fig:L-joint-mod}
\end{figure}

The NP-hardness proof uses the same gadgets as in the half-integral case,
with minor modifications: The variable gadget now consists of rectangles of length $8$,
and all paths to clauses are also formed by rectangles of at least this length.
The clause gadget is modified as illustrated in Figure~\ref{fig:clause-gadget-mod}.
It has three incoming paths, connecting variables to the clause, and the last rectangle
is of length $14$. The gadget is designed to satisfy the following properties 
similar to the Clause Lemma in Section~\ref{sec:NP}:
\begin{enumerate}
\item If all three incoming thin rectangles have their vertices in the back half,
      the projection cannot be completed to an embedding.
\item If the projection is half-integral, and at least one of the incoming rectangles
      has its vertex in the front pixel, we can place a vertex in the clause rectangle
      such that the orientation of all incident triangles is preserved.
\end{enumerate}

\begin{figure}
\begin{center}
\includegraphics[width=6cm]{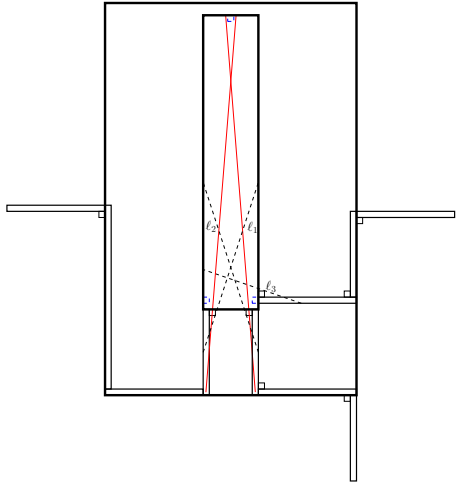}
\end{center}
\caption{The modified clause gadget: If the vertices of all three incoming rectangles are placed in
the back half, the projection of the clause triangle must be left of $\ell_1$, right of $\ell_2$
and below $\ell_3$, which is impossible. If the projection is integral and at least one of the
incoming rectangles has its vertex in the front pixel, we can choose the center of one of the dashed
pixels as projection for the clause rectangle.}
\label{fig:clause-gadget-mod}
\end{figure}
We can show that a formula encoded in a \texttt{grid3sat} instance is satisfiable
if and only if the corresponding partition has an embedding.
The proof is the same as for Lemma~\ref{lem:sat_emb_equi}: Given an embedding,
all vertices in a variable cycle must be either mapped to the front half or
the back half, and we set the variable to $1$ or $0$ accordingly. 
This is a satisfying assignment, because for any clause, one of the incoming 
rectangles must have its vertex in the front half, and applying the L-joint lemma
backwards ensures that the corresponding literal satisfies the clause.
Vice versa, given a satisfying assignment, we can construct a half-integral
projection in exactly the same way as in the proof of Lemma~\ref{lem:sat_emb_equi}
which yields an embedding.

\section{Details of the counterexample constructions}
\label{app:puzzle}

We have created configurations in Theorems~\ref{thm:beta_balanced_high}
and~\ref{thm:counter_beta_balanced_dd} which lead to a simplex whose orientation
is not preserved under the center projection. We have to show that those 
initial configurations can be completed to a $\beta$-balanced partition.
For Theorem~\ref{thm:beta_balanced_high}, this is relatively straight-forward.
For the proof, we need the following definition: 
Let $\partition$ be a partition in $\R^d$, $B$ a box in the partition with corners $(q_1,\ldots,q_D)$
(with $D=2^d$), 
and $p$ a point in the interior of $B$.
\emph{Splitting} $B$ at $p$ means to replace $B$ in $\partition$ with $D$ boxes defined by the corners
$p$ and $q_i$, with $1\leq i\leq D$.

\begin{figure}[hb]
\centering
\includegraphics[width=7.5cm]{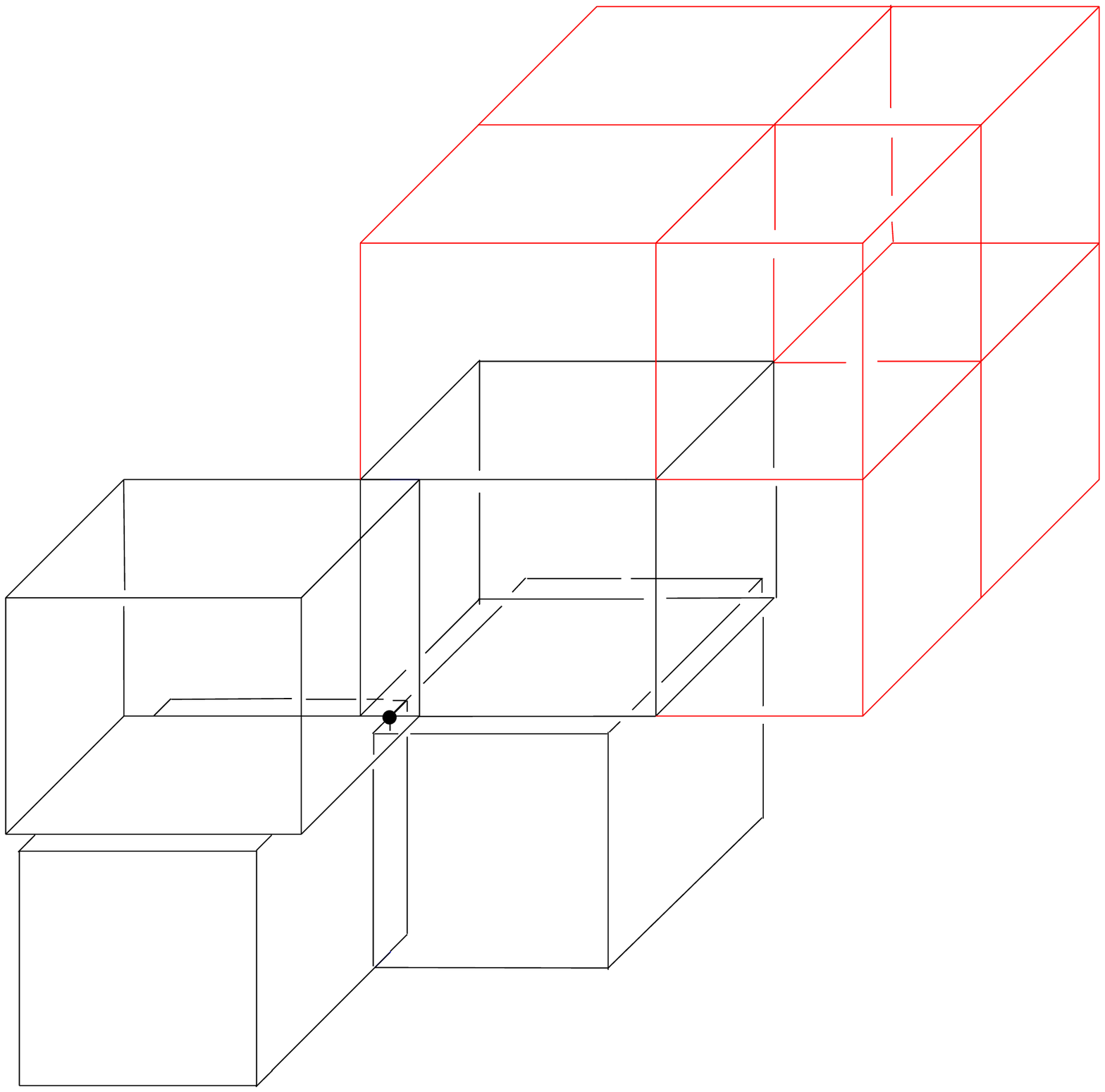}
\caption{A part of the complete partition, after the split of one of the $\tilde{B}_i$.}
\label{fig:completion}
\end{figure}

\begin{lemma}\label{lem:first_filling_lemma}
The partial partition $R_0,\ldots,R_3$ from Theorem~\ref{thm:beta_balanced_high}
can be completed to a $\beta$-balanced partition in $\R^3$.
\end{lemma}
\begin{proof}
Let $b\in\Z$ as in Theorem~\ref{thm:beta_balanced_high}, such that $\beta>\frac{b+2}{b-2}$.
We will construct a partition of $B:=[-2b,2b]^3$, consisting of exactly $64=8\cdot 8$ boxes, 
with $R_0,\ldots,R_3$ among them.
Fix a corner $a_i$ of $B$ (with $1\leq i\leq 8$). 
It lies in an octant of $\R^3$ with respect to the origin, and there
is exactly one box $R_j$ (with $0\leq j\leq 3$) that extends into this octant.
Let $c_j$ denote the corner of $R_j$ that is closest to the origin
(by construction, it is one unit length away from the origin).
Note that $c_j$ does not necessarily lie in the same octant as $a_i$.
We consider the box $\tilde{B}_i$ 
with (opposite) corners $a_i$ and $c_j$.
Observe that the side lengths of $\tilde{B}_i$ are $2b-1$, $2b$, or $2b+1$.
Also, $B_1,\ldots,B_8$ form a partition of $B$, and each $R_j$ (with $0\leq j\leq 3$)
lies in the corner of some $B_i$.

We split each $\tilde{B}_i$
into $8$ boxes: If $\tilde{B}_i$ contains some $R_j$ (with $0\leq j\leq 3$), 
we let $c^\ast_{j}$ denote the corner of $R_j$ opposite to $c_j$, and split $\tilde{B}_i$ at $c^\ast_j$;
clearly, one of the $8$ created boxes is $R_j$. See Figure~\ref{fig:completion} for an illustration.
If $\tilde{B}_i$ does not contain any $R_j$, we simply split at an integer point closest to its center.
It is simple to verify that all boxes created by these operations have side lengths of at least $b-2$
and at most $b+2$. Since $\frac{b+2}{b-2}$ by assumption, the rectangular partition is $\beta$-balanced,
but not center-embeddable, because it contains $R_0,\ldots,R_3$.
\end{proof}

For Theorem~\ref{thm:counter_beta_balanced_dd}, the completion of the initial configuration 
is technically more challenging than 
in Theorem~\ref{thm:beta_balanced_high} because we can only use 
square boxes to fill the empty space. Indeed, our construction needs much
more filling space to ``balance out'' the differences between the
square boxes, and it requires square boxes of many different size, all close to $b$.

The construction needs two preparatory lemmas. The first is a basic fact from elementary number theory

\begin{lemma}\label{lem:coprime_lemma}
Let $z_0:=1$ and $z_1,\ldots,z_k$ be the first $k$ prime numbers (starting with $z_1=2$).
Let $b\in\Z$ such that $b=\lambda\cdot (z_1\cdots z_k)+1$ with some $\lambda\in\Z$.
Then, the $(k+1)$ numbers
$b+z_0-1,\ldots,b+z_k-1$
are pairwise coprime.
\end{lemma}
\begin{proof}
Assume for a contradiction that $b+z_i-1$ and $b+z_j-1$ have a common prime factor $p$.
Then $p=z_\ell$ for $1\leq\ell\leq k$, because the two numbers differ by less than $z_k$.
Moreover, $b-1$ is a multiple of $z_\ell$. It follows that $z_\ell$ divides both $z_i$
and $z_j$, which is a contradiction.
\end{proof}

The next lemma proves that we can fill arbitrary
large boxes in $\R^d$ just using square boxes with $2^d$ pairwise coprime side lengths.

\begin{lemma}\label{lem:square_filling}
Let $D:=2^d$ and let $b_1<\ldots<b_{D}$ be pairwise coprime integers.
Then there exists some $L_0\in\Z$ such that every box with smallest side length at least $L_0$
can be partitioned into square boxes of side lengths $b_1,\ldots,b_D$. Moreover, the partition
can be chosen such that a box of side length $b_1$ is in one of its corners.
\end{lemma}
\begin{proof}
Set $B:=\{b_1,\ldots,b_{D}\}$.
Choose $L_0$ such that for any integer $L\geq L_0$, and any two
disjoint subsets $B_1,B_2$ of $B$, 
$L$ can be represented as 
\[L=\lambda_1 \prod_{b\in B_1}b + \lambda_2\prod_{b\in B_2}b\]
with $\lambda_1,\lambda_2$ non-negative numbers.
We prove inductively that any box with minimal side length at least $L_0$
can be partitioned. If $d=1$, the statement is trivial.
Let $B$ be such a box in $\R^d$ with $d>1$ 
and let $\tilde{B}$ denote the same
box projected into $\R^{d-1}$. Using the induction hypothesis with the numbers
$b_1,\ldots,b_{D/2}$, we can partition $\tilde{B}$ with square boxes in $\R^{d-1}$.
If we carry this partition into $\R^d$ with the lower side of every box being in
the plane $x_d=0$, this does not yield a box because
the square boxes have different heights in $x_d$-direction; however, we can stack up
boxes to match their heights. Formally, define $P_1=b_1\cdots b_{D/2}$
and place $P_1/b_i$ vertical copies for a square box of length $b_i$. In this way,
we construct a box whose projection to $\R^{d-1}$ is $\tilde{B}$, and whose length
in $x_d$-direction is $P_1$.
We can do exactly the same construction using the integers $b_{D/2+1},\ldots,b_{D}$,
which yields a box of length $P_2=b_{D/2+1}\cdots b_{D}$ in $x_d$-direction.
Let $L_d$ be the length of $B$ in $x_d$-direction.
By assumption, we can find non-negative integers $\lambda_1,\lambda_2$ with
$\lambda_1 P_1 + \lambda_2 P_2=L_d$.
Thus, placing $\lambda_1$ copies of the box of first type and $\lambda_2$ copies
of the box of second type gives the desired box.
\end{proof}

\begin{lemma}
The partial partition $R_0,\ldots,R_d$ from Theorem~\ref{thm:counter_beta_balanced_dd} can be completed
to a $\beta$-balanced cubical partition in $\R^3$.
\end{lemma}

\begin{proof}

The first step is to choose $a$ and $b$ appropriately. Note that the proof of 
Theorem~\ref{thm:planar_balancing} requires that $\beta>\frac{b}{a}\geq\frac{d}{d-2}$,
and such that $\det M^{(1)}_d$ is negative. 
We impose additional conditions:
Let $D:=2^d$ and $z_1,\ldots,z_{D-1}$ be the sequence of the first $D-1$ prime numbers
(note that the magnitude of $z_{D-1}$ only depends on $d$).
We choose $a$ and $b$ such that 
$b=\lambda z_1\cdots z_{D-1}+1$ for some $\lambda\in\Z$, such that
\[\beta>\frac{b+z_{D-1}-1}{a}>\frac{b}{a}\geq \frac{d}{d-2},\]
\noindent and such that $\det M^{(1)}_d$ is negative.
In particular, applying Lemma~\ref{lem:coprime_lemma} yields $D$ pairwise coprime integers
$b:=b_1<\ldots<b_D=b+z_{D-1}-1$.

Lemma~\ref{lem:square_filling} applied on the sequence $b_1,\ldots,b_D$ 
asserts the existence of some $L_0$ such that any box with side lengths at least $L_0$
can be filled. We choose some $L\geq L_0+1$ that is a multiple of $a$,
and set $B:=[-L,L]^d$. We have to show that we can construct a partition of $B$ 
containing $R_0,\ldots,R_d$
with square boxes of sizes $a,b_1,b_2,\ldots,b_D$ only. 
By our choice of $a$ and $b$, the so-obtained cubical
partition is $\beta$-balanced, and the claim is proven.

The construction works similar as
in Lemma~\ref{lem:first_filling_lemma}.
We first partition $B$
into $D$ boxes $\tilde{B}_j$ 
(which are not square boxes in general), each anchored at one corner of $B$,
and with $d+1$ of them containing one of the initial square boxes $R_i$.
Each $\tilde{B}_j$ has side lengths of at least $L-1\geq L_0$; therefore,
it can be filled with squares boxes of lengths $b_1,\ldots,b_D$ according to
Lemma~\ref{lem:square_filling}. If 
$\tilde{B}_j$ contains some $R_i$, we can enforce that $R_i$ is in the corner
of the filling, as required. A special case is that $\tilde{B}_j$
contains $R_0$ which is of length $a$. However, by construction,
this $\tilde{B}_j$ is a square box of length $L$, and we can simply fill
it with square boxes of length $a$ because $L$ is a multiple of $a$.
\end{proof}

\section{Determinant bound}
\label{app:determinant}

We prove in this section that for the $(d+1)\times (d+1)$ matrix
\[M^{(0)}_d:=\left(\begin{array}{cccccc}
1 & -\frac{a}{2} &  & \ldots &  & -\frac{a}{2} \\
1 & \frac{b}{2} & -\frac{b}{2} & \ldots & \ldots & -\frac{b}{2}\\
1 & -\frac{b}{2} & \frac{b}{2} & -\frac{b}{2} & \ldots & -\frac{b}{2}\\
\vdots & \vdots & \ddots & \ddots & \ddots & \vdots\\
1 & -\frac{b}{2} & \ldots & -\frac{b}{2} & \frac{b}{2} & -\frac{b}{2}\\
1 & -\frac{b}{2} & \ldots & \ldots & -\frac{b}{2} & \frac{b}{2}
\end{array}\right),\]
it holds that 
\[\det M^{(0)}_d=\frac{1}{2}b^{d-1}a - \frac{1}{2}(d-2)b^d=\frac{1}{2}b^{d-1}(da-(d-2)b).\]

Indeed, we subtract the $(i-1)$-st row from the $i$-th row for $i=n+1,\ldots,2$ to obtain
\[\det M^{(0)}_d=\det \left(\begin{array}{cccccc}
1 & -\frac{a}{2} &  & \ldots &  & -\frac{a}{2} \\
 & \frac{a+b}{2} & \frac{a-b}{2} & \ldots & \ldots & \frac{a-b}{2}\\
0 & -b & b &  &  & \\
0 &  & -b & b & & \\
\vdots & & & \ddots & \ddots &  \\
0 & & & & -b & b
\end{array}\right)=\det \left(\begin{array}{ccccc}
\frac{a+b}{2} & \frac{a-b}{2} & \ldots & \ldots & \frac{a-b}{2}\\
 -b & b &  &  & \\
  & -b & b & & \\
 & & \ddots & \ddots &  \\
& & & -b & b
\end{array}\right),\]
where empty spots are zero.
We simplify further by 
factoring out $1/2$ in the second row,
factoring out $b$ in rows $3,\ldots,d+1$,
and shifting the second row to the bottom, thereby changing the sign
of every row that goes up. With these steps, we obtain
\[
\det M^{(0)}_d=
\det \left(\begin{array}{ccccc}
\frac{a+b}{2} & \frac{a-b}{2} & \ldots & \ldots & \frac{a-b}{2}\\
 -b & b &  &  & \\
  & -b & b & & \\
 & & \ddots & \ddots &  \\
& & & -b & b
\end{array}\right)
=\frac{1}{2}b^{d-1}\det \left(\begin{array}{ccccc}
1 & -1 &  &  & \\
 & 1 & -1 & & \\
& & \ddots & \ddots &  \\
& & & 1 & -1\\
a+b & a-b & \ldots & \ldots & a-b\\

\end{array}\right).\]
It is straight-forward to verify that the determinant of the rightmost $(d\times d)$-matrix
equals the sum of its last row, which is $da-(d-2)b$.

\section{Details on quantifier elimination}
\label{app:qe}
The proof of Theorem~\ref{thm:cubical_partition_bound_3d} relies on the 
non-existence of planes that intersect certain point sets in $\R^3$.
We will give some details on how we prove this non-existence.
Recall from the proof that
we defined $a$ and $b$ to be the shortest and longest sides among the cubes $R_0,\ldots,R_3$.
By scaling, we can just assume that $a=1$. Because of the assumption, it
follows that $b\leq \beta < 3$. 

\paragraph{Regular case}
The regular configuration
\[\{0\}, \{1\}, \{2,3\}, \{4,5,6,7\},\]
induces the point sets
\begin{eqnarray*}
C_0^{(b)}&:=&\mathrm{hull}\{(-1,-1,-1),(-b,-b,-b)\}\\
C_1^{(b)}&:=&\mathrm{hull}\{(1,-1,-1),(b,-b,-b)\}\\
C_2^{(b)}&:=&\mathrm{hull}\{(-1,1,-1),(-b,b,-b),(1,1,-1),(b,b,-b)\}\\
         &  &-\mathrm{hull}\{(-1,1,-1),(-b,b,-b)\}-\mathrm{hull}\{(1,1,-1),(b,b,-b)\}\\
C_3^{(b)}&:=&\mathrm{hull}\{(-1,-1,1),(-b,-b,b),(-1,1,1),(-b,b,b),(1,-1,1),(b,-b,b),(1,1,1),(b,b,b)\}\\
         &  &-\mathrm{hull}\{(-1,-1,1),(-b,-b,b),(-1,1,1),(-b,b,b)\}\\
         &  &-\mathrm{hull}\{(-1,-1,1),(-b,-b,b),(1,-1,1),(b,-b,b)\}\\
         &  &-\mathrm{hull}\{(-1,1,1),(-b,b,b),(1,1,1),(b,b,b)\}\\
         &  &-\mathrm{hull}\{(1,-1,1),(b,-b,b),(1,1,1),(b,b,b)\},
\end{eqnarray*}
where $\mathrm{hull}$ stands for the convex hull of a point set.
Assume for a contradiction that 
$e$ is a plane intersecting $C_0^{(b)},\ldots,C_3^{(b)}$. 
Since $C_i^{(b)}\subset C_i^{(3)}$, $e$ intersects
$C_0^{(3)},\ldots,C_3^{(3)}$ as well. 
Let $e$ be defined by $t_0,\ldots,t_3\in\R$ via the equation
\[(t_0,t_1,t_2)\cdot(x,y,z)+t_3=0.\]
We first claim that $t_0\neq 0$. Indeed, if $t_0=0$, the projection of $e$ into the $yz$-plane
is a line. This line has to intersect the projections of the $C_i^{(3)}$. However, these projections
give the subsets of the rectangles defined in Figure~\ref{fig:beta_balanced_plane} (left),
and it can be seen with similar methods as in Theorem~\ref{thm:planar_balancing} that no
such line exists for $\beta<3$. Moreover, it can be easily verified that $e$ is not
parallel to any coordinate plane.

Since $t_0\neq 0$, we can normalize such that $t_0=1$. We define the induced mapping 
\[E:\R^3\rightarrow\R, (x,y,z)\mapsto (1,t_1,t_2)\cdot (x,y,z) + t_3.\]
Clearly, $e=E^{-1}(0)$, and
if for two points $p_1,p2\in\R^3$, $E(p_1)\cdot E(p_2)<0$, 
then $p_1$ and $p_2$ are separated in different halfspaces by $e$.
From this property, we can immediately deduce

\begin{lemma}\label{lem:formula_1}
$e$ intersects each $C_i^{(3)}$, $1\leq i\leq 4$ if and only if $E$ satisfies the following four formulas

\begin{eqnarray*}
E((-1,-1,-1))\cdot E((-3,-3,-3))\leq 0&&\\
E((1,-1,-1))\cdot E((3,-3,-3))\leq 0&&\\
E((-1,1,-1))\cdot E((3,3,-3))< 0&\vee& E((1,1,-1))\cdot E((-3,3,-3))< 0\\
(E((-1,-1,1))\cdot E((3,3,3))< 0 &\vee& E((1,-1,1))\cdot E((-3,3,3))< 0\\
&\vee&E((-1,1,1))\cdot E((3,-3,3))< 0\\
&\vee& E((1,1,1))\cdot E((-3,-3,3))< 0)
\end{eqnarray*}
\end{lemma}

\begin{proof}
The $\Leftarrow$-direction is simple to proof.
For the $\Rightarrow$-direction, note that the
first two statements are trivially satisfied. For the third statement, 
note first that $e$ cannot contain
the whole line segment from $(-1,1,-1)$ to $(1,1,-1)$, neither the line segment
from $(-3,3,-3)$ to $(3,3,-3)$ (because this would contradict $t_0\neq 0$).
Let $f$ be the plane
that contains $C_2^{(3)}$. The intersection of $e$ and $f$ is a line on $f$
that intersects the boundary of $C_2^{(3)}$ in exactly two points. It is therefore
clear that this line separates two opposite vertices in two halfplanes.
Thus, this pair is separated into two halfspaces by $e$ and this is precisely
what is checked by the third statement.

For the last statement, note that $e$ cannot completely contain
any boundary face of $C_3^{(3)}$, because this would mean that $e$ is
parallel to a coordinate plane which is easily checked to be impossible.
Therefore, $e$ intersects through the interior of $C_3^{(3)}$.
Therefore $e$ separates at least one of the four pairs of opposite vertices of the polytope,
and this is checked in the forth statement.
\end{proof}

With Lemma~\ref{lem:formula_1}, we can formulate a quantified formula in the real variables 
$t_1,t_2,t_3$ which is true if and only if a plane $e$ as required exists. However, using
a quantifier elimination program, we can easily compute that the formula is in fact false.

\paragraph{Singular case} This case is completely analogous. For completeness, we write down
the sets:
\begin{eqnarray*}
C_0^{(b)}&:=&\mathrm{hull}\{(-1,-1,-1),(-b,-b,-b),(1,-1,-1),(b,-b,-b)\}\\
         &  &-\mathrm{hull}\{(-1,-1,-1),(-b,-b,-b)\}-\mathrm{hull}\{(1,-1,-1),(b,-b,-b)\}\\
C_1^{(b)}&:=&\mathrm{hull}\{(-1,1,-1),(-b,b,-b),(1,1,-1),(b,b,-b)\}\\
         &  &-\mathrm{hull}\{(-1,1,-1),(-b,b,-b)\}-\mathrm{hull}\{(1,1,-1),(b,b,-b)\}\\
C_2^{(b)}&:=&\mathrm{hull}\{(-1,-1,1),(-b,-b,b),(-1,1,1),(-b,b,b)\}\\
         &  &-\mathrm{hull}\{(-1,-1,1),(-b,-b,b)\}-\mathrm{hull}\{(-1,1,1),(-b,b,b)\}\\
C_3^{(b)}&:=&\mathrm{hull}\{(1,-1,1),(b,-b,b),(1,1,1),(b,b,b)\}\\
         &  &-\mathrm{hull}\{(1,-1,1),(b,-b,b)\}-\mathrm{hull}\{(1,1,1),(b,b,b)\}.
\end{eqnarray*}

Defining $e$ and $E$ as above, we obtain:

\begin{lemma}\label{lem:formula_2}
$e$ intersects each $C_i^{(3)}$, $1\leq i\leq 4$ if and only if $E$ satisfies the following four formulas

\begin{eqnarray*}
E((-1,-1,-1))\cdot E((3,-3,-3))< 0&\vee& E((1,-1,-1))\cdot E((-3,-3,-3))< 0\\
E((-1,1,-1))\cdot E((3,3,-3))< 0&\vee& E((1,1,-1))\cdot E((-3,3,-3))< 0\\
E((-1,-1,1))\cdot E((-3,3,3))< 0&\vee& E((-1,1,1))\cdot E((-3,-3,3))< 0\\
E((1,-1,1))\cdot E((3,3,3))< 0&\vee& E((1,1,1))\cdot E((3,-3,3))< 0.
\end{eqnarray*}
\end{lemma}

Again, it can be proven with a quantifier elimination program that the induced quantified formula
is false.

\end{appendix}


\begin{thebibliography}{10}
\newcommand{\enquote}[1]{``#1''}
\providecommand{\url}[1]{\texttt{#1}}
\providecommand{\urlprefix}{URL }

\bibitem{ah-topologie}
P.~Alexandroff, H.~Hopf: \emph{Topologie I}.
\newblock Springer, 1935.

\bibitem{bek-computing}
P.~Bendich, H.~Edelsbrunner, M.~Kerber: \enquote{Computing Robustness and
  Persistence for Images}.
\newblock \emph{IEEE Transactions on Visualization and Computer Graphics}
  \textbf{16} (2010) 1251--1260.

\bibitem{bg-drawing}
M.~Bern, J.~R. Gilbert: \enquote{Drawing the planar dual}.
\newblock \emph{Inf. Process. Lett.} \textbf{43} (1992) 7--13.

\bibitem{ch-partial}
G.~E. Collins, H.~Hong: \enquote{Partial Cylindrical Algebraic Decomposition
  for quantifier elimination}.
\newblock \emph{Journal of Symbolic Computation} \textbf{12} (1991) 299--328.

\bibitem{ew-fixed}
P.~Eades, N.~Wormald: \enquote{Fixed Edge Length Graph Drawing is NP-hard}.
\newblock \emph{Discrete Applied Mathematics} \textbf{28} (1990) 111--134.

\bibitem{eh-computational}
H.~Edelsbrunner, J.~Harer: \emph{Computational Topology. An introduction}.
\newblock Amer.\ Math.\ Soc., 2010.

\bibitem{ek-freudenthal}
H.~Edelsbrunner, M.~Kerber: \enquote{Dual Complexes of Cubical Subdivisions in
  $\mathbb{R}^n$}.
\newblock \emph{Discrete and Computational Geometry} \textbf{47} (2012)
  393--414.

\bibitem{eppstein-regular}
D.~Eppstein: \enquote{Regular Labelings and Geometric Structures}.
\newblock \emph{CoRR} \textbf{abs/1007.0221} (2010).

\bibitem{ek-simultaneous}
C.~Erten, S.~G. Kobourov: \enquote{Simultaneous Embedding of a Planar Graph and
  Its Dual on the Grid}.
\newblock \emph{Theor. Comp. Sys.} \textbf{38} (2005) 313--327.

\bibitem{godau-difficulty}
M.~Godau: \enquote{On the difficulty of embedding planar graphs with
  inaccuracies}.
\newblock In: \emph{Graph Drawing}, \emph{LNCS}, vol. 894, 254--261, 1995.

\bibitem{hatcher}
A.~Hatcher: \emph{Algebraic Topology}.
\newblock Cambridge University Press, 2002.

\bibitem{he-finding}
X.~He: \enquote{On Finding the Rectangular Duals of Planar Triangular Graphs}.
\newblock \emph{SIAM J. Comput.} \textbf{22} (1993) 1218--1226.

\bibitem{lly-compact}
C.-C. Liao, H.-I. Lu, H.-C. Yen: \enquote{Compact floor-planning via orderly
  spanning trees}.
\newblock \emph{Journal of Algorithms} \textbf{48} (2003) 441 -- 451.

\bibitem{mps-rectangular}
S.~Muthukrishnan, V.~Poosala, T.~Suel: \enquote{On Rectangular Partitionings in
  Two Dimensions: Algorithms, Complexity and Applications}.
\newblock In: \emph{Database Theory --- ICDT'99}, \emph{LNCS}, vol. 1540,
  236--256, 1999.

\bibitem{reisz-statistical}
E.~Reisz: \enquote{The Rectangular Statistical Cartogram}.
\newblock \emph{Geographical Review} \textbf{24} (1934) 292--296.

\bibitem{wf-diamond}
K.~Weiss, L.~De~Floriani: \enquote{Diamond Hierarchies of Arbitrary Dimension}.
\newblock \emph{Computer Graphics Forum} \textbf{28} (2009) 1289--1300.

\end{thebibliography}
\end{document}